\PassOptionsToPackage{pdfpagelabels=false, pdfusetitle}{hyperref} 
\documentclass[twoside,11pt,final]{entics} 
\usepackage{enticsmacro}
\usepackage[T1]{fontenc}
\usepackage{microtype}
\usepackage{upgreek}
\usepackage{mathtools}
\usepackage{graphicx}
\usepackage{amssymb}
\usepackage{adjustbox} 
\usepackage{xspace}
\usepackage{csquotes}
\usepackage[inline]{enumitem} 
\usepackage{stmaryrd} 
\usepackage{thmtools} 
\usepackage{flushend} 
\usepackage{hyphenat} 
\usepackage{centernot} 

\usepackage{tikz}
\usetikzlibrary{positioning}
\usetikzlibrary{calc}
\usetikzlibrary{decorations.markings}
\usetikzlibrary{arrows.meta}
\usetikzlibrary{backgrounds}
\usetikzlibrary{fit}
\usetikzlibrary{intersections}

\newcommand*{\eg}{e.g.\@,\xspace}

\newcommand*{\cf}{cf.\@\xspace}
\newcommand*{\ie}{i.e.\@,\xspace}

\newcommand*{\resp}{resp.\@\xspace}

\makeatletter
\newcommand*{\etc}{%
	\@ifnextchar{.}%
	{etc}%
	{etc.\@\xspace}%
}
\makeatother

\def\mc#1{\mathcal{#1}}

\def\set#1{\{#1\}}
\def\bs{\backslash}

\def\Inf{\bigsqcap}
\def\Sup{\bigsqcup}
\def\down{\mathop{\downarrow}}
\def\up{\mathop{\uparrow}}

\def\comp{\mathrel{\frown}}
\def\ncomp{\not\mathrel{\mkern-3mu\comp}}

\def\dec#1{\lfloor#1\rfloor}

\def\sem#1{[\![#1]\!]}

\def\tuple#1{\langle{#1}\rangle}

\def\Conf{\mathbb{C}} 
\def\ES{\mathbb{E}}

\newcommand{\po}{\mathbb{P}}

\newcommand{\pr}{\mathsf{Pr}}

\def\K{\mathsf{K}}

\def\pred{\mathit{pred}} 
\def\max{\mathit{max}} 
\def\dec#1{\lfloor#1\rfloor} 
\def\Pfin{{\mc P}_{\!\!\mathit{fin}}} 

\def\eqOrb{\mathrel{\sim_\odot}}
\def\ort#1{\mathrel{\perp_{#1}}} 

\def\imp{\mathrel{\Rightarrow}}
\def\nil{\emptyset}

\def\cf{\mathrel{\#}}
\def\Del{\mathrel{\vartriangle}}
\def\grpAct{\mathrel{\odot}}
\def\monAct{\mathrel{\cdot}}
\def\Del{\mathrel{\vartriangle}}
\def\grpAct{\mathrel{\odot}}
\def\Orb{\mathit{Orb}}

\def\0{\mathbf{0}}
\def\1{\mathbf{1}}

\definecolor{clem}{rgb}{0.0, 0.72, 0.92}

\definecolor{jk}{rgb}{0.5, 0.1, 0.42}

\makeatletter
\@namedef{df*}#1{\par\begingroup\def\proofname{#1}\pf\endgroup\ignorespaces}
\expandafter\let\csname enddf*\endcsname=\enddf
\makeatother
\newenvironment{definition*}{\begin{df*}}{\end{df*}}

\newtheorem{property}{Property}

\def\equationautorefname~#1\null{(#1)\null} 

\makeatletter
\providecommand*{\toclevel@theorem}{0}
\providecommand*{\toclevel@proposition}{0}
\providecommand*{\toclevel@definition}{0}
\providecommand*{\toclevel@notation}{0}
\providecommand*{\toclevel@example}{0}
\providecommand*{\toclevel@description}{0}
\providecommand*{\toclevel@lemma}{0}
\providecommand*{\toclevel@remark}{0}
\providecommand*{\toclevel@corollary}{0}
\providecommand*{\toclevel@property}{0}
\providecommand*{\toclevel@conjecture}{0}
\makeatother

\colorlet{fillcolor}{gray!20}
\colorlet{fillborder}{gray!80}

\tikzset{
	smap/.style={%
		<->,
		>={Straight Barb[scale=0.8]},
		line width=1pt
	},
	ifont/.style={%
		font=\fontsize{14}{14}\selectfont
	},
	umap/.style={%
		->,
		>={Straight Barb[scale=0.8]},
		line width=1pt
	},
	mapl/.style={%
		midway,
		fill=white,
		inner sep = .05em
	},
	mmap/.style={%
		dash dot,
		->,
	},
	dmap/.style={%
		dashed,
		->,
		>={Straight Barb[scale=0.8]},
	},
	cf/.style={%
		densely dotted,%
		line width=.7pt,%
		draw,%
		postaction={%
			decorate,%
			decoration={%
				markings,%
				mark=at position 0.5 with {%
					\node[midway, fill=white, rounded corners=1.5, inner sep=.7pt]{\(\cf\)};}
			}
		}
	},
	ord/.style={%
		densely dotted,%
		line width=.7pt,%
		draw,%
		postaction={%
			decorate,%
			decoration={%
				markings,%
				mark=at position 0.5 with {%
					\node (symb) [midway, fill=white, rounded corners=1.5, inner sep=.7pt]{\(<\)};}
			}
		}
	},
	ordf/.style={%
		densely dotted,%
		line width=.7pt,%
		draw,%
		postaction={%
			decorate,%
			decoration={%
				markings,%
				mark=at position 0.5 with {%
					\node (symb) [midway, fill=fillcolor, rounded corners=1.5, inner sep=.7pt]{\(<\)};}
			}
		}
	},
	su/.style={%
		->,
		>={Straight Barb[scale=0.8]},
		line width=.5pt
	},
	conf/.style = {
		sloped,
		node contents={\#},
		fill=white,
		inner xsep=1pt,
		inner ysep=0pt
	}
}

\newcommand{\rOrbitConf}{
	\node (rempt) {\(\emptyset\)};
	\node (a) [above left = .3cm and .6cm of rempt] {\(\{a\}\)};
	\node (b) [above right  = .3cm and .6cm of rempt] {\(\{c\}\)};
	\node (c) [above = .5cm of rempt] {\(\{b\}\)};
	\node (ac) [above of = a] {\(\{a, b\}\)};
	\node (bc) [above of = b] {\(\{b, c\}\)};
	\node (abc) [above = 1.1cm of c] {\(\{a, b, c\}\)};
	\draw [su] (rempt) -- (a);
	\draw [su] (rempt) -- (b);
	\draw [su] (rempt) -- (c);
	\draw [su] (a) -- (ac);
	\draw [su] (b) -- (bc);
	\draw [su] (c) -- (ac);
	\draw [su] (c) -- (bc);
	\draw [su] (ac) -- (abc);
	\draw [su] (bc) -- (abc);
}

\newcommand{\lOrbitConf}{
	\node (lempt) {\(\emptyset\)};
	\node (a) [above left = .3cm and .6cm of lempt] {\(\{a\}\)};
	\node (b) [above right  = .3cm and .6cm of lempt] {\(\{b\}\)};
	\node (c) [above = .5cm of lempt] {\(\{c\}\)};
	\node (ac) [above of = a] {\(\{a, c\}\)};
	\node (bc) [above of = b] {\(\{b, c\}\)};
	\node (abc) [above = 1.1cm of c] {\(\{a, b, c\}\)};
	\draw [su] (lempt) -- (a);
	\draw [su] (lempt) -- (b);
	\draw [su] (lempt) -- (c);
	\draw [su] (a) -- (ac);
	\draw [su] (b) -- (bc);
	\draw [su] (c) -- (ac);
	\draw [su] (c) -- (bc);
	\draw [su] (ac) -- (abc);
	\draw [su] (bc) -- (abc);
}

\newcommand{\bOrbitConf}{	
	\node (bempt) {\(\emptyset\)};
	\node (a) [above left = .3cm and .6cm of bempt] {\(\{a\}\)};
	\node (b) [above right  = .3cm and .6cm of bempt] {\(\{b\}\)};
	\node (c) [above = 1.1cm of bempt] {\(\{a, b\}\)};
	\node (ac) [above of = a] {\(\{a, c\}\)};
	\node (bc) [above of = b] {\(\{b, c\}\)};
	\node (abc) [above = .5cm of c] {\(\{a, b, c\}\)};
	\draw [su] (bempt) -- (a);
	\draw [su] (bempt) -- (b);
	\draw [su] (a) -- (ac);
	\draw [su] (b) -- (bc);
	\draw [su] (a) -- (c);
	\draw [su] (b) -- (c);
	\draw [su] (ac) -- (abc);
	\draw [su] (bc) -- (abc);
	\draw [su] (c) -- (abc);
}

\newcommand{\tOrbitConf}{
	\node (tempt) {\(\emptyset\)};
	\node (a) [above left = .3cm and .6cm of tempt] {\(\{a\}\)};
	\node (b) [above right  = .3cm and .6cm of tempt] {\(\{c\}\)};
	\node (c) [above  = .3cm of tempt] {\(\{b\}\)};
	\node (ac) [above of = a] {\(\{a, b\}\)};
	\node (bc) [above of = b] {\(\{b, c\}\)};
	\node (abc) [above = .5cm of c] {\(\{a, c\}\)};
	\draw [su] (tempt) -- (a);
	\draw [su] (tempt) -- (b);
	\draw [su] (tempt) -- (c);
	\draw [su] (a) -- (abc);
	\draw [su] (a) -- (ac);
	\draw [su] (b) -- (abc);
	\draw [su] (b) -- (bc);
	\draw [su] (c) -- (ac);
	\draw [su] (c) -- (bc);
}

\newcommand{\lOrbitConfAfterB}{
	\node (empt) {};
	\node (a) [above left = .3cm and .6cm of empt] {};
	\node (b) [above right  = .3cm and .6cm of empt] {\(\emptyset\)};
	\node (c) [above = .5cm of empt] {};
	\node (ac) [above of = a] {};
	\node (bc) [above of = b] {\(\{c\}\)};
	\node (abc) [above = 1.5cm of c] {\(\{a, c\}\)};
	\draw [su] (b) -- (bc);
	\draw [su] (bc) -- (abc);
}

\newcommand{\bOrbitConfAfterB}{	
	\node (empt) {};
	\node (a) [above left = .3cm and .6cm of empt] {};
	\node (b) [above right  = .3cm and .6cm of empt] {\(\emptyset\)};
	\node (c) [above = 1.1cm of empt] {\(\{a\}\)};
	\node (ac) [above of = a] {};
	\node (bc) [above of = b] {\(\{c\}\)};
	\node (abc) [above = .5cm of c] {\(\{a, c\}\)};
	\draw [su] (b) -- (bc);
	\draw [su] (b) -- (c);
	\draw [su] (bc) -- (abc);
	\draw [su] (c) -- (abc);
}

\newcommand{\tOrbitConfAfterB}{
	\node (empt) {};
	\node (a) [above left = .3cm and .6cm of empt] {};
	\node (b) [above right  = .3cm and .6cm of empt] {};
	\node (c) [above  = .3cm of empt] {\(\emptyset\)};
	\node (ac) [above of = a] {\(\{a\}\)};
	\node (bc) [above of = b] {\(\{c\}\)};
	\node (abc) [above = .5cm of c] {};
	\draw [su] (c) -- (ac);
	\draw [su] (c) -- (bc);
}

\newcommand{\recCS}{
	\node (a1) [above left = .3cm of empt1] {\(\{a\}\)};
	\node (b1) [above right  = .3cm of empt1] {\(\{b\}\)};
	\node (bc1) [above right = .3cm of b1] {\(\{b, c\}\)};
	\node (ab1) at (bc1 -| empt1) {\(\{a, b\}\)};
	\node (bcd1) [above = .3cm of bc1] {\(\{b, c, d\}\)};
	\node (CS1) [below = .2cm of empt1] {\(\Conf\)};
	\draw [su] (empt1) -- (a1);
	\draw [su] (empt1) -- (b1);
	\draw [su] (a1) -- (ab1);
	\draw [su] (b1) -- (ab1);
	\draw [su] (b1) -- (bc1);
	\draw [su] (bc1) -- (bcd1);
}

\newcommand{\recsCS}{
	\node (d2) [above left = .3cm of empt2] {\(\{d\}\)};
	\node (c2) [above = .3cm of empt2] {\(\{c\}\)};
	\node (ac2) [above left = .3cm of c2] {\(\{a, c\}\)};
	\node (bc2) [above right = .3cm of c2] {\(\{b, c\}\)};
	\node (abc2) at (bcd1 -| empt2) {\(\{a, b, c\}\)};
	\node (CS2) [below = .2cm of empt2] {\(\set{b, c} \grpAct \Conf\)};
	
	\draw [su] (empt2) -- (d2);
	\draw [su] (empt2) -- (c2);
	\draw [su] (c2) -- (ac2);
	\draw [su] (c2) -- (bc2);
	\draw [su] (ac2) -- (abc2);
	\draw [su] (bc2) -- (abc2);
}

\newcommand{\recES}{
	\node (a) [left  = .8cm of marq1] {$\set{a}$};
	\node (b) [right = .8cm of marq1] {\(\{b\}\)};
	\node (bc) [above = .7cm of marq1] {\(\{b, c\}\)};
	\node (bcd) [above = .8cm of bc] {\(\{b, c, d\}\)};
	\node (ES1) [below = .3cm of marq1] {\(\ES\)};
	\draw [cf] (a) -- (bc);
	\draw [ordf] (b) -- (bc) node[midway, xshift=9pt, yshift=20pt] {$X$};
	\draw [ord] (bc) -- (bcd);
	
	\begin{scope}[on background layer]
		\node [fill=fillcolor, fit=(b) (bc), rectangle, rotate fit = 45, rounded corners=.35cm, inner sep=9pt, inner xsep = -18pt, draw=fillborder] {};
	\end{scope}
}

\newcommand{\recsES}{
	\node (a2) [left = .7cm of marq2] {$\set{a}$};
	\node (b2) [right = .7cm of marq2] {\(\{b\}\)};
	\node (bc2) [above = .7cm of marq2] {\(\{b, c\}\)};
	\node (bcd2) [above = .8cm of bc2] {\(\{b, c, d\}\)};
	\node (ES2) [below = .3cm of marq2] {\(X \grpAct \ES
		\)};
	\draw [cf] (bc2) -- (bcd2);
	\draw [ord] (bc2) -- (b2);
	\draw [ord] (bc2) -- (a2);	
}

	\newcommand{\rOrbit}{
	\begin{tikzpicture}
		\rOrbitConf
	\end{tikzpicture}
}

\newcommand{\tOrbit}{
	\begin{tikzpicture}
		\tOrbitConf
	\end{tikzpicture}
}

\newcommand{\lOrbit}{
	\begin{tikzpicture}
		\lOrbitConf
	\end{tikzpicture}
}

\newcommand{\bOrbit}{
	\begin{tikzpicture}
		\bOrbitConf
	\end{tikzpicture}
}    

\newenvironment{proofwq}{\begin{pf}}{\end{pf}} 

\def\conf{MFPS 2025} 	
\volume{5}			
\def\volu{5}			
\def\papno{2}			
\def\mydoi{16677}
\def\lastname{Aubert, Krivine} 
\def\runtitle{Reversible computations are computations} 
\def\copynames{C. Aubert, J. Krivine}    
\begin{document}

\begin{frontmatter}
  \title{Reversible Computations are Computations}
  \author{Clément Aubert\thanksref{a}\thanksref{clememail}\thanksref{NSF}}	
   \author{Jean Krivine\thanksref{b}\thanksref{jeanemail}}		
   \address[a]{Augusta University\\
   	Augusta, GA, USA}
   \thanks[clememail]{Email: \href{mailto:caubert@augusta.edu} {\texttt{\normalshape
        caubert@augusta.edu}}} 
  \thanks[NSF]{This work has benefited from the support from the National Science Foundation under \href{https://www.nsf.gov/awardsearch/showAward?AWD_ID=2242786}{Grant No. 2242786}. 
  }
  \address[b]{CNRS, Univ. Paris Cité, IRIF\\
 				Paris, France} 
   \thanks[jeanemail]{Email: \href{mailto:jean.krivine@irif.fr} {\texttt{\normalshape
			jean.krivine@irif.fr}}} 


\begin{abstract} 
		Causality serves as an abstract notion of time for concurrent systems. A computation is causal, or simply valid, if each observation of a computation event is preceded by the observation of its causes. The present work establishes that this simple requirement is equally relevant when the occurrence of an event is invertible. We propose a conservative extension of causal models for concurrency that accommodates reversible computations. We first model reversible computations using a symmetric residuation operation in the general model of configuration structures. We show that stable configuration structures, which correspond to prime algebraic domains, remain stable under the action of this residuation. We then derive a semantics of reversible computations for prime event structures, which is shown to coincide with a switch operation that dualizes conflict and causality.
\end{abstract}
\end{frontmatter}
\section{Introduction}

In physics, the concept of time being \enquote{reversible} stems from the fact that many fundamental physical laws are symmetric with respect to time. This means that the equations governing certain physical processes remain unchanged whether time moves forward or backward. For example, applying Newton's laws of motion to a celestial body does not reveal whether the body is moving forward or backward in time, as the equations validate the trajectory in both directions. The present work explores this intriguing idea in the context of models of concurrent systems, where time is abstracted to a weaker notion of causality or precedence. Specifically, it answers positively the following question: \emph{Are standard models of concurrent systems also models of reversible ones?}

This is not 
a rhetorical question, as unlike time in physics, computations can, in practice, be reversed. Reversible computations have been studied in a variety of contexts~\cite{rc2020}, such as minimising entropy~\cite{BuhTroVit01,She_etal03,Abr05}, resolving deadlocked transactions~\cite{Pra87,DanKri05}, and designing abstract machines and software capable of debugging concurrent execution traces~\cite{Lie_etal12,Lan_etal18}. With these objectives, abstract language models for reversible computing have been developed using process algebraic frameworks~\cite{DanKri04,Lan_etal10,PhiUli06,DomKriVar13,Lan_etal21,Aub24} or P/T nets~\cite{DanKriSob06,Bar_etal18,Mel_etal21}.

This paper proposes to study such reversible calculi from a denotational perspective, \eg focusing on relating computation events rather than producing them. 

Providing a denotational semantics of reversible concurrent programs has been studied over the last 10 years, in a line of research that considers backward events as first class citizens. It follows that the causal structure of computations induced by backward events is modeled by specific relations, such as reverse causality~\cite{PhiUli15,Uli_etal18,Mel_etal21}.

In this work, we aim to determine whether the concept of reversible computations necessitates disrupting the state-of-the art models of concurrency or can instead be seamlessly integrated. To achieve this, we examine event structure semantics for concurrent systems within the same framework established by Winskel and others since the 1980s~\cite{BouCas89,NiePloWin81,Win82}. 

We base the present work on a simple and general model of concurrent systems called configuration structures, where a process is denoted by the configurations (the collection of observable events) it can produce. This approach, which expresses the combination of observations that are allowed or disallowed as a partial order of configurations, was foundational to Winskel's pioneering work and has been extensively studied in subsequent research~\cite{Bal_etal17,GlaPlo95,GlaPlo09}.

After 
introducing configuration structures and their semantics (\autoref{sec:conf-structures}), we show that reversible computations can be understood as a (partial) group action, which we call symmetric residuation, of configuration structures on themselves (\autoref{sec:group-action}). This generalizes the notion of residuals, a standard method of retrieving a transition system from a denotational model~\cite{BouCas87,BouCas89}, while remaining purely within the simple partial order theory. Configuration structures, in their full generality, can model concurrent systems with arbitrary causal and conflict structures. While these structures may not always correspond to the computation of a traditional algebraic program like the calculus of communicating systems (CCS~\cite{Mil89}), a well-known property of these structures is that those satisfying certain stability axioms correspond to prime algebraic domains. These domains provide a semantics for CCS-like calculi \emph{via} prime event structures~\cite{Win82}.  The first main result of this paper is to show that \textbf{stable configuration structures remain stable under symmetric residuation} (\autoref{thm:stable}), which is a guarantee that symmetric residuation produces configuration structures that are also prime event structures. On these premises, the second main contribution of this paper (\autoref{thm:switch}) is to show that \textbf{transporting symmetric residuation to (prime) event structures corresponds to a switch operation} on the graph that represents the causality and conflict relation. The term switch is intended to draw a connection with a class of well-known graph isomorphisms called Seidel switches \cite{Seidel73} which is discussed in 
conclusion 
 (\autoref{sec:ccl}).

Importantly, 
we show that the 
backward or forward orientation of an event exists only relatively to a given past state:
like in our example of the celestial body in motion, our approach yields a causal semantics for reversible computation in which the notion of \enquote{backward moves} only appears when configuration structures are equipped with a distinguished configuration that we call \emph{referential} (Definition~\ref{def:pointed-structure}). We show that these \textbf{pointed configuration structures correspond to a form of polarized event structures}, 
where negative events come and go under the action of our switching operation (\autoref{thm:pol-switch}), providing an adequate model 
for the explicit orientation of reversible process calculi.

\section{Configuration structures}\label{sec:conf-structures}
We recall the definition of general configuration structures as presented for instance in 
\cite[Definition 1.1]{GlaPlo09}.

\begin{definition}
\label{def:configuration-structure}
	A \emph{configuration structure} is a pair $\Conf=\tuple{E,X}$ where $E$ is a countable set of computation \emph{events} (ranged over by lower case letters \(a\) through \(f\)) and $X\subseteq \mc P(E)$ is a set of \emph{configurations} (ranged over by lower case letters \(x\), \(y\) and \(z\)). 
	A configuration structure is \emph{rooted} when $\emptyset\in X$.
\end{definition}

For such configurations, we use the direct notation $x\in \Conf$ when $x\in X$.
In the remainder of the paper we only consider rooted configurations, and use $\mc C_E$ to denote the set of configuration structures over $E$.

Configurations can be partially ordered by set inclusion. Thus, for a given configuration structure $\Conf=\tuple{E,X}$ and for any $Y\subseteq X$, we can use the standard order theoretic notations:
\begin{align*}
	{\down}_\Conf Y & \coloneqq \set{x \in \Conf \mid\exists y\in Y: x\subseteq y}\\
	{\up}^\Conf Y & \coloneqq \set{x \in \Conf \mid \exists y\in Y:y\subseteq x}
\end{align*}
We also write $x=\Sup^\Conf Y$ whenever $Y$ has a least upper bound $x \in \Conf$. Similarly $x=\Inf_\Conf Y$ denotes the greatest lower bound of $Y$ when it exists in $\Conf$. Lastly, for any two configurations $x,y\in \Conf$, we write $x\comp_\Conf y$ whenever $\Sup^\Conf\set{x,y}$ exists and say that \(x\) and \(y\) are \emph{compatible in $\Conf$}.
Two configuration structures are \emph{equivalent} if there exists a bijection on events that preserves inclusion of configurations and that equates them.

\begin{example}\label{ex:configuration-structures}
	We give below three examples of configuration structures of $\mc C_{\set{a,b,c}}$, where an arrow denotes inclusion. Notice that $\set{a}\comp_{\Conf_0} \set{b}$, $\set{b}\comp_{\Conf_0}\set{c}$ and $\set{a}\comp_{\Conf_0}\set{c}$ but $\Sup^{\Conf_0}\set{a,b,c}$ is not defined, which models a ternary conflict where at most two events can trigger. Configuration $\Conf_1$ models a form of disjunctive causality: the event $c$ can be observed provided event $a$ or event $b$ (or both) appeared first. Lastly, configuration $\Conf_2$ models the fact that events $a$ and $b$ are not compatible unless $c$ is triggered.
	
	\noindent
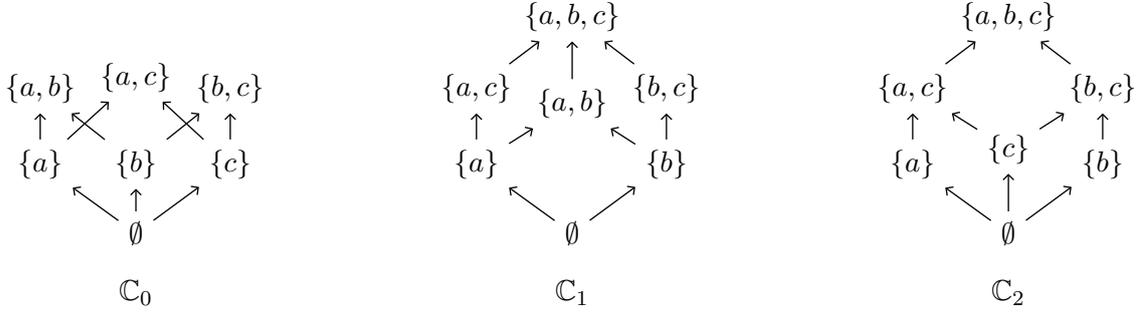
{
	\centering
	
	\begin{tikzpicture}
		\tOrbitConf
		\node (c) [below = .2cm of tempt] {\(\Conf_0\)};
		\begin{scope}[xshift = 15em]
			\bOrbitConf
			\node (c) [below = .2cm of bempt] {\(\Conf_1\)};
		\end{scope}
		\begin{scope}[xshift = 30em]
			\lOrbitConf
			\node (c) [below = .2cm of lempt] {\(\Conf_2\)};
		\end{scope}
	\end{tikzpicture}

}
\end{example}

Configuration structures can be equipped with an operation of \emph{residuation}~\cite{BouCas87}, which describes which configurations remain reachable once a set of events has been triggered. 

\begin{definition}
\label{def:residuation}
	Let  $\Conf\in\mc C_E$.  For all finite $x\in \Conf$, we define the \emph{residual} of $\Conf$ \emph{after} $x$:
	\[
	x\monAct\Conf \coloneqq \tuple{E,\set{z\in\mc P(E)\mid \exists y \in {\up}^\Conf\set{x}: z=y\bs x}}
	\]
	where $y\bs x \coloneqq\set{a\in y\mid a\not\in x}$ is the classical set difference.
\end{definition}

Since $(\mc P(E),\cup)$ is a commutative monoid, residuation can be understood as the action of a (partial) monoid on itself: 

\begin{proposition}[Monoid action]\label{prop:monoid-action}
	The operator $(\monAct):\Pfin(E)\times\mc C_E \to \mc C_E$ is a monoid action on configuration structures, \ie:
	\begin{itemize}
		\item for all configurations $x$, $y$,  if $y\in \Conf$ and $x\in y\cdot \Conf$, then $x\monAct (y\monAct \Conf)=(x\cup y)\monAct\Conf$.
		\item $\emptyset\monAct\Conf = \Conf$.
	\end{itemize}
\end{proposition}
\begin{proof}
	Each configuration $z\in x \cdot (y \cdot \Conf)$ is of the form $z = (z'\bs y) \bs x=z'\bs (x\cup y)$ for some $z'\in \Conf$.  So $z\in x \cdot (y\cdot \Conf)$ is equivalent to $z\in (x\cup y)\cdot\Conf$. 
	The fact that $\emptyset\cdot\Conf=\Conf$ is a direct consequence of the definition of the residuation operation.
\end{proof}

\begin{example}\label{ex:residuation}
	We give below the residuals of the configurations of Example~\ref{ex:configuration-structures} after $\set{b}$. Notice that each configuration residual can still trigger the events $a$ and $c$, but they are incompatible in $\set{b}\cdot\Conf_0$, independent in $\set{b}\cdot\Conf_1$ and sequential in $\set{b}\cdot\Conf_2$.
	
	\noindent
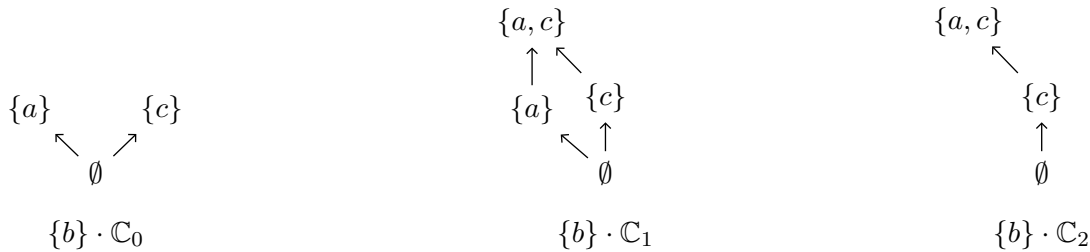
{
	\centering
	
	\begin{tikzpicture}
		\tOrbitConfAfterB	
		\node (c) [below = .2cm of c] {\(\set{b}\cdot\Conf_0\)};
		\begin{scope}[xshift = 15em]
		\bOrbitConfAfterB
		\node (c) [below = .2cm of b] {\(\set{b}\cdot\Conf_1\)};
	\end{scope}
		\begin{scope}[xshift = 30em]
		\lOrbitConfAfterB
		\node (c) [below = .2cm of b] {\(\set{b}\cdot\Conf_2\)};
	\end{scope}
	\end{tikzpicture}

}
\end{example}


For a Configuration structure $\Conf\coloneqq (E,X)$, consider  $\mc R(\Conf)\coloneqq \set{\Conf'\in\mc C_E\mid \exists x\in\Conf:\Conf'=x\monAct\Conf}$ the \emph{reachable} residuals of $\Conf$.
Given an \emph{initial} configuration structure $\Conf_0$, one can build the \emph{labelled transition system}  $\mc T\coloneqq(\Conf_0, \to, \mc R(\Conf_0))$ where $\to\subseteq \mc R(\Conf_0)\times \mc P(E)\times \mc R(\Conf_0)$ is the transition relation defined as $\Conf\to_x\Conf'$ if and only if $\Conf'=x\monAct\Conf$. This models program states as configuration structures $\Conf_0,\Conf_1,\dots$ and computations as sequences of transition between them. Notice that since $\to$ is engendered by a monoid action, the transition relation is closed by composition (see Proposition~\ref{prop:monoid-action}). It is also equipped with causal and conflict informations (which is sometimes called a truly concurrent semantics) since the transitions $\Conf\to_x\Conf'$ and $\Conf\to_y\Conf''$ are concurrent if $x\comp_\Conf y$ and conflicting otherwise.

\section{Symmetric residuation}\label{sec:group-action}
\subsection{Building up intuitions}
To model reversible computation we wish to keep interpreting program states as regular configurations structures: sets of events ordered by inclusion. It would be tempting to consider introducing negative events as first class citizens, but this would force one to consider configurations modulo some equivalence $\set{a,a^-}\sim\emptyset$, which would entail unnecessary complications as set union is no longer associative under this equivalence (think of $(\set{a}\cup\set{a})\cup\set{a^-}$ and $\set{a}\cup(\set{a}\cup\set{a^-})$). 

The alternative is to work with a different residuation operation which yields a group action instead of a monoid one. To illustrate this, consider a system that can perform $a$ or $b$, but not both, and then terminates. This is modeled by the configuration structure $\Conf=\tuple{\set{a,b},\set{\emptyset,\set{a},\set{b},\set{a,b}}}$. Classical residuation yields $\Conf\to_{\set{a}}\tuple{\set{a,b},\set{\emptyset}}$, as the event $b$ is no longer visible in the future of $a\monAct\Conf$.
\def\ul#1{\underline{#1}}
Making this computations reversible requires a memory or special markers \cite{Lan_etal21} in order to keep track of past events. Following this intuition, we will introduce a reversible residuation operator, which takes care of producing the "memory" of having consumed $a$ and produce transitions of the form: 
$$
\tuple{\set{a,b},\set{\emptyset,\set{a},\set{b},\set{a,b}}}\to_{\set{a}}\tuple{\set{a,b},\set{\emptyset,\set{\ul a},\set{\ul a,b}}}
$$ 
whose target state reads as "one can trigger event $a$ (from the memory), and then one can trigger $b$". We will see that this transition can be derived using a symmetric version of the classical residuation, which uses symmetric set difference instead of set difference (Definition~\ref{def:sym-residuation}), but before diving into formal definitions, we can make a few more comments on the configuration structure $\tuple{\set{a,b},\set{\emptyset,\set{\ul a},\set{\ul a,b}}}$. First, the underlined events are pure annotations, in order to keep track of the fact that the configuration $\set{\ul a}$ contains a past event. As all occurrences of $a$ are underlined, there is no possible clash between a "forward" $a$ and a "backward" one. This annotation must be understood in a similar way as tracking a redex in the $\lambda$-calculus: useful for proving convergence, but having no impact on the semantics (see for instance Ref~\cite{Sel08}, Section 2.4). A consequence is that symmetric residuation on configuration structures will have a constant support: in our example $E=\set{a,b}$ throughout computations, although underlined events will vary. Another important aspect is that symmetric residuation is conservative: erasing configurations with an underlined component corresponds to performing classical residuation. Lastly, we will see that symmetric residuation is a group action on configuration structure, with events being their own inverse. In our example, we will have:
$$
\tuple{\set{a,b},\set{\emptyset,\set{a},\set{b},\set{a,b}}}\to_{\set{a}}\tuple{\set{a,b},\set{\emptyset,\set{\ul a},\set{\ul a,b}}}
\to_{\set{\ul a}} \tuple{\set{a,b},\set{\emptyset,\set{a},\set{b},\set{a,b}}}
$$ 
which comes back to the initial structure.

\subsection{Formally}

In order to make these intuitions formal, we first introduce an operation of \emph{symmetric residuation}, a generalization of the residuation of Definition~\ref{def:residuation}, which takes care of producing the causal structure of reversible computations. We introduce explicit backward (the underlined events of our example in the previous section) and forward events in \autoref{sec:pointed}.

Recall that $\Del$, the symmetric set difference, is defined as 
\(x\Del y \coloneqq  (x\cup y)\bs (x\cap y)\).

\begin{property}\label{prop:abelian}
	Any power set equipped with the symmetric difference forms an abelian group.
\end{property}

In particular $(\mc P(E),\Del)$ has an abelian group structure, where every event is its own inverse. Its action on configuration structures generates orbits that provide the mathematical foundations on which we will model reversible computations.

\begin{definition}
\label{def:sym-residuation}
	Let  $\Conf\in\mc C_E$.  For all finite $x\in\Conf$, we define the \emph{symmetric residual of $\Conf$ after $x$}:
	\[	x\grpAct\Conf \coloneqq\tuple{E,\set{z\in\mc P(E)\mid \exists y\in\Conf: z=y\Del x}}\text{.}\]
\end{definition}

\begin{proposition}[Group action]
	\label{prop:gp-act}
	The operator $(\grpAct):\Pfin(E)\times\mc C_E \to \mc C_E$ is a \emph{group action} on configuration structures, i.e.\@:
	\begin{itemize}
		\item for all finite configurations $x$, $y$, if \(x\in \Conf\) and $y\in x\grpAct\Conf$, then \(x\grpAct (y\grpAct \Conf) =(x\Del y)\grpAct\Conf\).
		\item $\emptyset\grpAct\Conf = \Conf$.
	\end{itemize}
\end{proposition}

\begin{proof}
	We need to show that $x\Del y$ is a configuration of $\Conf$. If $x\grpAct (y\grpAct \Conf)$ is defined, then $y$ is a configuration of $\Conf$ and $x$ is a configuration of $y\grpAct\Conf$. Now all configurations of $y\grpAct\Conf$ must be of the form $y\Del z$ for some configuration $z$ of $\Conf$. Therefore $x=y\Del z$ and thus, $x\Del y = (y\Del z)\Del y = z\Del (y\Del y)=z$ using Proposition~\ref{prop:abelian}. It entails that $x\Del y$ is indeed a configuration of $\Conf$.
	
	The fact that $\emptyset\grpAct\Conf = \Conf$ is a consequence of $\emptyset$ being the neutral element for the group $(\mc P(E),\Del)$.
\end{proof}

Notice that symmetric residuation $x\grpAct\Conf$ acts uniformly on all the configurations of $\Conf$. This is in contrast with the classical residuation $x\monAct\Conf$ (Definition~\ref{def:residuation}), which discards configurations of $\Conf$ that are not above $x$. Yet both symmetric and classical residuations coincide on configurations above $x$.

\begin{proposition}[Conservative extension]
	Let $\Conf\in\mc C_E$.
	For all finite $x \in \Conf$, if \(z \in x \monAct \Conf\) then \(z \in x \grpAct \Conf\).
\end{proposition}

\begin{proof}
	Each configuration \(z\) in \(x \monAct \Conf\) is such that \(\exists y\), \(x \subseteq y\) and \(z = y \bs x\).
	Since \(x \subseteq y\), \(y \bs x = (y \cup x) \bs (x \cap y) = x \Del y\) and hence \(z = x \Del y \in x \grpAct \Conf\).
\end{proof}

\begin{definition}

	The \emph{orbit} $\Orb(\Conf)$ of a configuration structure $\Conf\in\mc C_E$ is defined as:
	\[\Orb(\Conf)\coloneqq \set{\Conf'\in\mc C_E\mid \Conf'=x\grpAct\Conf\hbox{ for some finite $x\in\Conf$}}\text{.}\]
	The associated equivalence relation ${}\eqOrb{}\subseteq {\mc C_E\times \mc C_E}$ is:
	\(\Conf\sim_\odot\Conf' \hbox{ if and only if } \Conf' \in\Orb(\Conf)\).
\end{definition}

\begin{example}
	\label{ex:orbit}
	The configuration structures of Example~\ref{ex:configuration-structures} are equivalent by symmetric residuation, \ie $\Conf_0\eqOrb \Conf_1\eqOrb \Conf_2$, as pictured in \autoref{fig:orbit}, with \eg $\odot_{\set{a}}:y\mapsto \set{a} \Del y$.
\end{example}

\begin{figure}
{
	
\centering
\begin{tikzpicture}
	\def\radius{4cm}
	\draw[
	draw = none, 
	name path=c] (0:\radius) arc (0:360:\radius);
	\node[name path=top] at (90:\radius) (top) {\tOrbit}; 
	\node[name path=left] at (0:\radius) (right) {\lOrbit}; 
	\node[name path=right] at (180:\radius) (left) {\rOrbit}; 
	\node[name path=bottom] at (-90:\radius) (bottom) {\bOrbit}; 
	
	\path[name intersections={of=top and c, by={tl, tr}}]; 
	\path[name intersections={of=left and c, by={lt, lb}}]; 
	\path[name intersections={of=right and c, by={rt, rb}}]; 
	\path[name intersections={of=bottom and c, by={bl, br}}]; 
	
	\draw[smap] (tr) to[bend left=12] node[mapl, ifont]{\(\odot_{\set{c}}\)} (lt);
	\draw[smap] (tl) to[bend right=12] node[mapl, ifont]{\(\odot_{\set{b}}\)} (rt);
	\draw[smap] (rb) to[bend right=12] node[mapl, ifont]{\(\odot_{\set{a}}\)} (bl);
	\draw[smap] (lb) to[bend left=12] node[mapl, ifont]{\(\odot_{\set{a, b, c}}\)} (br);
	\draw[smap] (bottom.north) to [bend left] node[pos=.6, fill=white, inner sep=.05em, ifont]{\(\odot_{\set{a, b}}\)}
		(top.south);
	\draw[smap] (left.east) to [bend right] node[pos = .6, fill=white, inner sep=.05em, ifont]{\(\odot_{\set{b, c}}\)} (right.west);
\end{tikzpicture}

}
	\caption{Orbit of \(\Conf_0\)}\label{fig:orbit}
\end{figure}

Symmetric residuation does not act freely on configuration structures, \ie not all the group fixed points are trivial. Note, for instance, that the configuration structure $(E,\mc P(E))$ is invariant under the action of $\grpAct$.

\begin{example}
	\label{ex:fixed-point}
	The configuration structure $\Conf$ below remains invariant after symmetric residuation by $\set{b}$ (greyed out). Dashed arrows denote the function $\odot_{\set{b}}:y\mapsto \{b\} \Del y$.
	
	
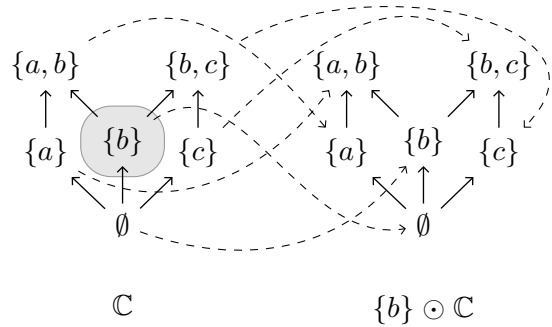
\noindent{

\centering

		\begin{tikzpicture}
			\node (empt1) {\(\emptyset\)};
			\node (a1) [above left  = 1.3em of empt1] {\(\{a\}\)};
			\node (b1) [above  = 1.3em of empt1] {\(\{b\}\)};
			\node (c1) [above right  = 1.3em of empt1] {\(\{c\}\)};
			\node(ab1) [above = .5cm of a1] {$\set{a,b}$};
			\node(bc1) [above = .5cm of c1] {$\set{b,c}$};
			\node(Conf1)[below  = 1.3em of empt1, node distance=0.55cm] {$\Conf$};
			\draw [su] (empt1) -- (a1);
			\draw [su] (empt1) -- (b1);
			\draw [su] (empt1) -- (c1);
			\draw [su] (a1) -- (ab1);
			\draw [su] (b1) -- (bc1);
			\draw [su] (c1) -- (bc1);
			\draw [su] (b1) -- (ab1);	
			
			\node (empt) [right  = 1.3em of empt1, xshift=3cm]{\(\emptyset\)};
			\node (a) [above left  = 1.3em of empt] {\(\{a\}\)};
			\node (b) [above  = 1.3em of empt] {\(\{b\}\)};
			\node (c) [above right  = 1.3em of empt] {\(\{c\}\)};
			\node(ab) [above = .5cm of a] {$\set{a,b}$};
			\node(bc) [above = .5cm of c] {$\set{b,c}$};
			\node(Conf)[below  = 1.3em of empt, node distance=0.55cm] {$\set{b}\grpAct\Conf$};
			
			\draw [su] (empt) -- (a);
			\draw [su] (empt) -- (b);
			\draw [su] (empt) -- (c);
			\draw [su] (a) -- (ab);
			\draw [su] (b) -- (bc);
			\draw [su] (c) -- (bc);
			\draw [su] (b) -- (ab);	
			
			\draw[dmap] (a1) to[out=330, in=235] (ab);
			\draw[dmap] (empt1) to[out=-20, in=235]  (b);
			\draw[dmap] (b1) to[out=35, in=190] (empt);
			\draw[dmap] (c1) to[out=45, in=140]  (bc);
			\draw[dmap] (ab1) to[out=30, in=130, looseness=1.15] (a);
			\draw[dmap] (bc1) to[out=30, in=45, looseness=1.3]  (c);
			\begin{scope}[on background layer]
				\node [fill=fillcolor, fit=(b1), rounded corners=.4cm, inner sep=4pt, draw=fillborder] {};
			\end{scope}	
		\end{tikzpicture}

}
\end{example}

\section{Pointed configuration structures}\label{sec:pointed}

Symmetric residuation yields the bare causal structure of reversible computations which is \enquote{direction insensible} in the sense that there is no way to tell whether an event is backward or forward in a computation. As a consequence, residuation by configurations that are neither enabling nor preventing any event from occurring are fixed points. This is, for instance, the case of configuration $\set{b}$ in Example~\ref{ex:fixed-point}.

Classical semantics of reversible computations distinguish backward and forward events~\cite{Graversen2018} or transitions~\cite{lmcs:12332}. In operational semantics in particular, it may be desirable to observe the direction of a transition to prioritize \enquote{forward} executions, so that backtracking occurs only when necessary~\cite{Lan_etal11}. Debugging concurrent executions also has a natural orientation, where backtracking is possible only during trace analysis~\cite{Lan_etal18}. Bisimulations for reversible concurrent processes usually require matching forward and backward transitions with transitions of the same direction~\cite{DeNMonVaa90,MonPis97,PhiUli12,Aubert2020b,Lanese2021,Aubert2025}. 

To comply with these semantics, we show here how symmetric residuation yields a natural notion of negative and positive configurations. Notice that adding a notion of direction of computation does not change the semantics we introduced so far, but rather annotates it with additional information. We do this by introducing \emph{pointed} configuration structures, which are just plain configuration structures equipped with a distinguished configuration--called its \emph{referential}--that points to the past.
A similar idea was used to draw an operational correspondence between trajectories in a configuration structure and a reversible operator algebra using labels~\cite{Aubert2020b}.	

\begin{definition}
\label{def:pointed-structure}
	A \emph{pointed configuration structure} \(\Conf[x^\dagger]\) is a pair $ \tuple{\Conf, x^{\dagger}}$ where $\Conf$ is a configuration structure and $x^\dagger\in\Conf$ is a finite configuration, which we call its \emph{referential}. A pointed configuration structure is \emph{initial} when $x^\dagger = \emptyset$.
\end{definition}

In the following, we simply write \(\Conf[x^{\dagger}] \in \mc C_E\) if \(\Conf \in \mc C_E\).

\begin{definition}
\label{def:pointed-residuation} 
	Let $\Conf[x^\dagger]\in\mc C_E$. For all finite $y\in\Conf$, we define
\(y\grpAct\Conf[x^\dagger] \coloneqq (y\grpAct\Conf)[y\Del x^\dagger]\).
\end{definition}

Symmetric residuation acts freely on pointed configuration structures, \ie the orbit of a pointed configuration structure has as many elements as the structure has configurations.

\begin{proposition}[Free action]\label{prop:free-action}
	Let $\Conf[x^\dagger]\in\mc C_E$. For all finite $y,z\in\Conf$, 
	\(y\grpAct\Conf[x^\dagger] = z\grpAct\Conf[x^\dagger]\) if and only if \(y = z\). 
\end{proposition}
\begin{proof}
	The only if direction is trivial. For the if part, it suffices to remark that the referential in $y\grpAct\Conf[x^\dagger]$ must be equal to the referential in $z\grpAct\Conf[x^\dagger]$ and thus, $y\Del x^\dagger = z\Del x^\dagger$. Since $\odot_x:y\mapsto x\Del y$ is injective we can conclude that $y=z$. 
\end{proof}

\begin{definition}
\label{def:orientation}
	Let $\Conf[x^\dagger]\in\mc C_E$. For all $a\in E$, say that $a$ is \emph{negative in $\Conf[x^\dagger]$} when $a\in x^\dagger$. It is \emph{positive in $\Conf[x^\dagger]$} otherwise.
	A configuration is \emph{backward} (\resp \emph{forward}) if all its events are \emph{negative} (\resp \emph{forward}). 
\end{definition}

Observe that referentials are finite configurations, and thus only finite configurations may be purely backwards. This echoes the \emph{well-foundedness} axiom~\cite[Definition 3.1]{LPU24} requiring reversible systems to have a finite past.

\begin{example}\label{ex:pointed-structure}
	Once equipped with a referential, the configuration structure of Example~\ref{ex:fixed-point} is no longer a fixed point, as illustrated in \autoref{fig:pointed-structure}, where negative signs denote negative events. 
\end{example}

\begin{figure}
	\noindent\centering
\begin{tikzpicture}
	\node (empt1) {\(\emptyset^\dagger\)};
	\node (a1) [above left = 1.3em of empt1] {\(\{a\}\)};
	\node (b1) [above = 1.3em of empt1] {\(\{b\}\)};
	\node (c1) [above right = 1.3em of empt1] {\(\{c\}\)};
	\node(ab1) [above = 1.3em of a1] {$\set{a,b}$};
	\node(bc1) [above = 1.3em of c1] {$\set{b,c}$};
	\node(Conf1)[below of = empt1, node distance=0.55cm] {$\Conf[\emptyset^\dagger]$};
	\draw [su] (empt1) -- (a1);
	\draw [su] (empt1) -- (b1);
	\draw [su] (empt1) -- (c1);
	\draw [su] (a1) -- (ab1);
	\draw [su] (b1) -- (bc1);
	\draw [su] (c1) -- (bc1);
	\draw [su] (b1) -- (ab1);	
	
	\node (empt) [right = 1.3em of empt1, xshift=3cm]{\(\emptyset\)};
	\node (a) [above left = 1.3em of empt] {\(\{a\}\)};
	\node (b) [above = 1.3em of empt] {\(\{b^-\}^\dagger\)};
	\node (c) [above right = 1.3em of empt] {\(\{c\}\)};
	\node(ab) [above = 1.3em of a] {$\set{a,b^-}$};
	\node(bc) [above = 1.3em of c] {$\set{b^-,c}$};
	\node(Conf)[below of = empt, node distance=0.55cm] {$\set{b} \grpAct \Conf[\emptyset^\dagger] = (\set{b}\grpAct\Conf)[\set{b}^\dagger]$};
	
	\draw [su] (empt) -- (a);
	\draw [su] (empt) -- (b);
	\draw [su] (empt) -- (c);
	\draw [su] (a) -- (ab);
	\draw [su] (b) -- (bc);
	\draw [su] (c) -- (bc);
	\draw [su] (b) -- (ab);	
	
	\begin{scope}[on background layer]
		\node [fill=fillcolor, fit=(b1), rounded corners=.4cm, inner sep=4pt, draw=fillborder] {};
	\end{scope}	
	
\end{tikzpicture}
	\caption{Symmetric residuation applied to a pointed configuration structure. The negative signs label the events (here $b$) that belong to the referential configuration.}\label{fig:pointed-structure}
\end{figure}
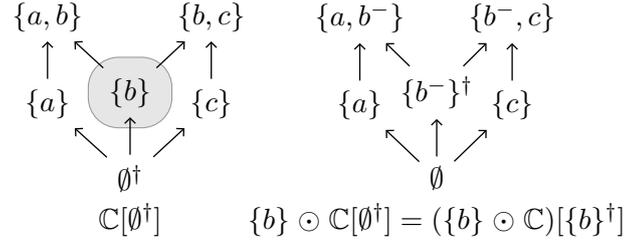

The referential of a pointed configuration structure provides a path back to its origin, which can be retrieved by performing the residuation $x^\dagger\grpAct \Conf[x^\dagger]$. In addition, Proposition~\ref{prop:free-action} guarantees that this is the only way to go back to the initial structure.

\section{Stable orbits and prime event structures}\label{sec:stable-orbit}
Configuration structures are \emph{extensional} models for concurrency, which means that the partial order of configurations is \emph{a priori} not reducible to relations on events themselves. For instance, in Example~\ref{ex:configuration-structures}, the fact that $\set{a,b}\ncomp_{\Conf_0}\set{a,c}$ is not a consequence of event $b$ being in conflict with event $c$ since $\set{b,c}$ is indeed a reachable configuration of $\Conf_0$. Prime event structures, which we introduce now, are models of concurrent systems whose set of admissible configurations is precisely engendered by a causality and a conflict relation on events. In particular, algebraic process calculi in the style of CCS~\cite{Mil89}, which are commonly used to model concurrent programs, can be interpreted as \emph{prime event structures}~\cite{Win82,NiePloWin81,BouCas89,GlaPlo09}. 

\begin{definition}
\label{def:PES}
	A (prime) event structure is a tuple $\ES=\tuple{E,<,\cf}$ where ${}<{}\subseteq E\times E$ is 
	a partial order (\eg transitive, anti-symmetric and reflexive) called the \emph{causality} relation and ${}\cf{}\subseteq E\times E$ an irreflexive symmetric \emph{conflict} relation satisfying:
	\begin{gather} 
		e \cf e' < e'' \imp e \cf e'' \tag{Principle of conflict heredity}\label{eq:pch}
	\end{gather}
	The \emph{configurations} of $\ES$, written $\mathit{Conf}(\ES)$, are the subsets of $E$ that are conflict-free and downward closed for the causality relation. Two prime event structures $(E,<,\cf)$ and $(E',<',\cf')$ are \emph{isomorphic} if there is a bijection $\phi:E\to E'$ that preserves causality and conflict, \ie $e<e'$ if and only if $\phi(e)<'\phi(e')$ and $e\cf e'$ if and only if $\phi(e)\cf'\phi(e')$.
\end{definition}

%
	%
			%
	
	A well-known result by Winskel~\cite{Win82} is that the configurations of a prime event structure coincide with \emph{stable} configuration structures.
	\begin{definition}
\label{def:stable}
		Let $\Conf\in\mc C_E$ be a configuration structure. Say that $\Conf$ is \emph{stable} if it is:
		\begin{description}
			\item[Rooted:] $\nil\in \Conf$
			\item[Connected:] for all $x\in \Conf$, $x\neq\nil\imp\exists a\in x:x\bs\set{a}\in\Conf$
			\item[Closed under bounded union:] for all $x,y,z\in\Conf$, $x\cup y\subseteq z \imp x\cup y\in \Conf$
			\item[Closed under intersection:] for all $x,y\in \Conf$, $x\cap y\in\Conf$
			\item[Coherent:] for all $x,y,z \in\Conf$, if there exists $z', z'', z'''\in \Conf$ such that $x\cup y\subseteq z'$, $y\cup z\subseteq z''$ and $x\cup z\subseteq z'''$, then $x\cup y\cup z\in \Conf$.
		\end{description}	
	\end{definition}
	
	For instance, in Example~\ref{ex:configuration-structures}, none of the configurations are stable:  $\Conf_0$ is missing $\set{a,b,c}$ to be coherent, $\Conf_1$ is missing $\set{c}$ to be closed under intersection, and $\Conf_2$ is missing $\set{a,b}$ to be closed under bounded union. As a consequence there are no prime event structures whose set of configurations is isomorphic to any of these structures, as indicated by the following theorem:
	
	\begin{theorem}[Winskel~\cite{Win82}]\label{thm:ConfOfES}
		Let $\ES=\tuple{E,<,\cf}$ be a prime event structure. The configuration structure $\mathsf{C}(\ES)\coloneqq \tuple{E,\mathit{Conf}(\ES)}$ is stable (see an illustration in \autoref{fig:es-cs}).
	\end{theorem}
	
	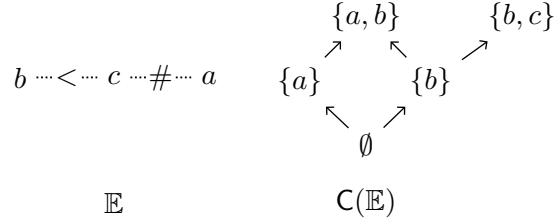
\begin{figure}
		\begin{center}	
	\begin{tikzpicture}
		\node(c) {\(c\)};
		\node(b) [left = .8cm of c]{\(b\)};
		\node(a) [right = .8cm of c]{\(a\)};
		\draw[cf] (c) -- (a);
		\draw[ord] (b) -- (c);

		\node (empt) [below right = 1em and 7.5em of c]{\(\emptyset\)};
		\node (a2) [above left = .3cm of empt] {\(\{a\}\)};
		\node (b2) [above right  = .3cm of empt] {\(\{b\}\)};
		\node (bc2) [above right = .3cm of b2] {\(\{b, c\}\)};
		\node (ab2) at (bc2 -| empt) {\(\{a, b\}\)};
		\node (caption2) [below = .4em of empt] {\(\mathsf{C}(\ES)\)};

		\node (caption1) at (caption2 -| c) {\(\ES\)};
		\draw [su] (empt) -- (a2);
		\draw [su] (empt) -- (b2);
		\draw [su] (a2) -- (ab2);
		\draw [su] (b2) -- (ab2);
		\draw [su] (b2) -- (bc2);
		
	\end{tikzpicture}
\end{center}
		\caption{An event structure and the set of its admissible configurations, forming a configuration structure.}\label{fig:es-cs}
	\end{figure}
	
	%
	%
	
	In the following, we let $\mc S_E\subseteq\mc C_E$ denote stable configuration structures over $E$. Stable configuration structures are, in fact, concrete representations of \emph{prime algebraic domains}, which we now introduce and which play an important role in connecting sets of configurations to event structures.
	
	\begin{definition}
\label{def:prime}
		Let $\po=(X,\sqsubseteq)$ be a partial order.
		A set $Y\subseteq X$ is:
		\begin{description}
			\item[Consistent] when $\Sup^\po Y$ exists.
			\item[Pairwise consistent] if \(Y \neq \emptyset\) and all its pairs of distinct elements form a consistent set.
			\item[Directed] when $Y\neq\emptyset$ and $\Sup^\po Y\in Y$.
		\end{description} 
		An element $x\in X$ is:
		\begin{description}
			\item[Compact] if for every directed $Y\subseteq X$ such that $x\sqsubseteq \bigsqcup^\po Y$, there exists $y\in Y$ such that $x\sqsubseteq y$. 
			\item[Complete prime] if for every pairwise consistent $Y\subseteq X$ such that  $x\sqsubseteq \bigsqcup^\po Y$, there exists $y\in Y$ such that $x\sqsubseteq y$. 
		\end{description} 
		Let $\K(\po)\subseteq X$ denote the \emph{compact elements} of $\po$ and $\pr(\po)\subseteq\K(\po)$ the (complete) prime elements of $\po$. We use $p,q,\dots$ instead of $x,y,\dots$ to highlight prime elements. Let $\dec{x}_\po\coloneqq \down_\po\set{x}\cap\pr(\po)$ denote the complete prime elements that are below $x$ in $\po$. 
	\end{definition}
	
	\begin{definition}
\label{def:prime-alg}
		Let $\po=(X,\sqsubseteq)$ be a partial order. It is:
		\begin{description}
			\item[Finitary] whenever for all $x\in \K(\po)$, $\down_\po\set{x}$ is finite. 	
			\item[Coherent] whenever for all pairwise consistent set $Y\subseteq X$, $\bigsqcup^\po Y$ exists.
			\item[Prime algebraic]  if for all \(x \in \po\), \(x = {\bigsqcup}^\po \dec{x}_\po \).
		\end{description}
		We call \emph{(prime algebraic) domains}, the partial orders that are finitary, coherent and prime algebraic. 
	\end{definition}
	
	\begin{theorem}[Winskel~\cite{Win82}]
		If a configuration structure $\Conf=(E,X)$ is stable then $\po=(X,\subseteq)$ is a domain of configurations.
	\end{theorem}
	
	\begin{theorem}[Winskel~\cite{Win82}]\label{thm:ESofConf}
		Let $\Conf\in\mc S_E$. Then  
		\(\mathsf{E}(\Conf)\coloneqq\tuple{\pr(\Conf),<,\cf}\)
		where $p<q$ if $p\subseteq q$ and $p\cf q$ if $p\ncomp_\Conf q$ is a prime event structure.
	\end{theorem}
	
	A fundamental result is that combining the functors $\mathsf{C}$ and $\mathsf{E}$ of Theorems~\ref{thm:ConfOfES} and \ref{thm:ESofConf} actually defines an adjunction, which is also an equivalence between stable configurations and prime event structures~\cite{Win82}: for all stable $\Conf$ and all prime event structure $\ES$: $\mathsf{C}(\mathsf{E}(\Conf)) \sim \Conf$ and $\mathsf{E}(\mathsf{C}(\ES))\sim\ES$. We can utilize this equivalence to transport symmetric residuation to prime event structures, but in order to do this, we must first verify that symmetric residuation preserves stability, \ie the properties of Definition~\ref{def:stable}. Observe that the orbit presented in Example~\ref{ex:orbit} does not preserve prime algebraicity since the topmost configuration structure is prime algebraic (but not coherent), when the other elements of the orbit are not.
	
	\begin{restatable}[Stable orbits]{theorem}{stableorbit}\label{thm:stable}
		The group action $\grpAct$ preserves stability, and is therefore a group action on stable configuration structures.
	\end{restatable}
	
	This property serves as a cornerstone result, enabling the definition of an event-structure semantics for reversible computations.
	This result, combined with the adjunction that makes prime event structures equivalent to stable configuration structures, allows one to conclude that symmetric residuation operates on prime event structures as well. \autoref{sec:prime-preserving} is devoted to tracking the effect of symmetric residuation on prime elements. In the meantime, we show that Winskel's construction on pointed stable configuration structures yields prime event structures that are equipped with a polarity map attributing a sign to events.
	
	\section{Polarized event structures}
	
	The construction of \autoref{thm:ESofConf} maps configuration structures to prime event structures using their complete prime elements. We extend this construction to build \emph{polarized event structures} out of pointed configuration structures. Polarized event structures are just plain prime event structures, in which events have a sign which characterizes whether they have a forward or backward contribution to the configurations.
	
	\begin{definition}
 \label{def:polarized-event-structure}
		A \emph{polarized event structure} $\ES[\pi]$ is a pair \(\tuple{\ES, \pi}\) where \(\ES =(E,<,\cf)\) is an event structure and $\pi:E\to\set{-1,1}$, the \emph{polarity of $\ES[\pi]$}, satisfies 
		\(\mathit{Neg}(\ES[\pi])\coloneqq\set{e\in E\mid \pi(e)<0}\in\mathit{Conf}(\ES)\). 
		Two polarized event structures $\ES[\pi]$ and $\ES'[\pi']$ are isomorphic if $\ES$ and $\ES'$ are isomorphic, and the induced bijection $\phi:E\to E'$ satisfies \(\pi(e) = \pi'(\phi(e))\).
	\end{definition}
	
	That \(\mathit{Neg}(\ES[\pi])\) is a configuration of \(\ES[\pi]\) guarantees that past events have indeed been produced by a forward computation, a standard requirement in reversible semantics~\cite[Definition 3.1]{LPU24}.
	
	\begin{proposition}\label{prop:polarized-winksel}
		Let $\Conf[x^\dagger]$ be a pointed stable configuration structure and define $\mathsf{E}(\Conf[x^\dagger]):=\mathsf{E}(\Conf)[\pi_{x^\dagger}]$ where for all $p\in\pr(\Conf)$, $\pi_{x^\dagger}(p)<0$ if $p\subseteq x^\dagger$ and $\pi_{x^\dagger}(p)>0$ otherwise. Then $\mathsf{E}(\Conf[x^\dagger])$ is a polarized event structure.
	\end{proposition}
	\begin{proof}
		Suffices to show that $N\coloneqq\mathit{Neg}(\mathsf{E}(\Conf)[\pi_{x^\dagger}])$ is a configuration of $\mathsf{E}(\Conf)$. We have $N=\set{p\in\pr(\Conf)\mid p\subseteq x^\dagger}$. Thus all the events of $N$ have an upper bound in $\Conf$ and are therefore not in conflict. $N$ is also downwards $<$-closed by construction and is therefore a configuration of $\mathsf{E}(\Conf)$. 
	\end{proof}
	
	\section{Mapping prime elements through symmetric residuation}\label{sec:prime-preserving}
	
	Symmetric residuation describes how a new configuration structure $x\grpAct\Conf$ is obtained from a configuration structure $\Conf$, and can take the form of a map $\odot_x:\Conf\to x\grpAct\Conf$ where, for all $y\in\Conf$, $\odot_x(y)=x\Del y$ (see an illustration on Example~\ref{ex:fixed-point}). However, this does not map prime elements of its source to prime elements of its targets, take for instance $\odot_{\set{b}}(\set{a})=\set{a,b}$ in Example~\ref{ex:fixed-point}. This would lead to a complex definition of the action of symmetric residuation on prime event structures, where events would be erased and created during computations. 
	
	For all stable configuration structure $\Conf$ and all finite configuration $x\in\Conf$, we establish now the existence of a bijective map $\sigma_x:\pr(\Conf)\to\pr(x\grpAct\Conf)$ (Definition~\ref{def:res-map}) which allows us to show that prime elements are indeed preserved by symmetric residuation, although their causal relationship is transformed (\autoref{thm:effect}). We first need to state some useful properties of prime elements in stable configuration structures.
	
	In prime algebraic partial orders, complete primes have a notoriously simpler characterization in terms of the unicity of their immediate predecessor (see Baldan \emph{et al.}~\cite[Lemma 13]{Bal_etal17}).
	
	\begin{proposition}[Unique predecessor]\label{prop:Bal_etal17}
		\label{prop:prime-complete-pred}
		Given a prime algebraic partial order \(\po\), an element \(x \in\po\) is a complete prime if and only if it has a unique immediate predecessor, \ie there exists $y\in\po$ such that
		\[
		\pred_\po(x)\coloneqq \max_\po({\down}_\po\set{x}\bs\set{x})=\set{y}
		\]
	\end{proposition}
	
	\begin{proof}
		Suppose by contradiction that \(p \in \po\) is prime but that it has two different immediate predecessors \(y\) and \(y' \in \po\).
		Then, \(p \sqsubseteq \bigsqcup \{y, y'\}\) and yet there is no \(z \in \{y, y'\}\) such that \(p \sqsubseteq z\), contradicting the prime completeness of \(p\).
		
		Conversely, suppose \(y \in X\) has a unique immediate predecessor \(y'\). Since $\po$ is prime algebraic we have $y=\Sup^\po\dec{y}_\po$. By contradiction, suppose that $y\not\in\pr(\po)$, we would have $\Sup^\po\dec{y}_\po = \Sup^\po\dec{y'}_\po=y'$. This would contradict $y'\subsetneq y$ and therefore $y$ must be a complete prime.
	\end{proof}
	
	\def\der#1{\updelta_{#1}}
	\begin{definition}

		Let $\Conf\in\mc S_E$ be stable configuration structure and $p\in\pr(\Conf)$. We define the \emph{derivative} of $p$ in $\Conf$  as $\der\Conf(p)\coloneqq a$ whenever $p\bs\pred_\Conf(p)=\set{a}$.
	\end{definition}
	
	\begin{proposition}[Event introduction]\label{prop:prime-prop}
		Let $\Conf\in\mc S_E$. For all $a\in E$ and $x\in\Conf$, $a\in x$ if only if there exists $p\in\pr(\Conf)$ such that $\der\Conf(p)=a$ and $p\subseteq x$. We say that $p$ \emph{introduces} $a$ in $\Conf$.
	\end{proposition}
	\begin{proof}
		The only if part is trivial since $\der\Conf(p)$ implies $a\in p$ and together with  $p\subseteq x$ it implies $a\in x$. For the other part, let $x\in\Conf$ such that $a\in x$. Since $\Conf$ is prime algebraic, $x=\Sup^\Conf\dec{x}_\Conf$ which is also the union of all prime elements below $x$ (by stability of $\Conf$). This implies that the set $P_a\coloneqq\set{p\in \dec{x}_\Conf\mid a\in p}$ is not empty. Take the minimal elements $\mathit{min}(P_a)$ of $P_a$ and suppose, by contradiction that it is not a singleton.  
		We would have $p\in \mathit{min}(P_a)$ and $q\in\mathit{min}(P_a)$ with $a\in p\cap q=:y$ and $p\neq q$. Now $y$ is a configuration of $\Conf$ (by stability) which is either prime and contains $a$ or has prime containing $a$ below it. This would contradict the minimality of $\mathit{min}(P_a)$. Let $\mathit{min}(P_a)=\set{p}$. By Proposition~\ref{prop:Bal_etal17}, $p$ has a single immediate predecessor $z$ such that $a\not\in z$ and $\der\Conf(p)=a$.
	\end{proof}
	
	\begin{proposition}[Unicity of event introducer]\label{prop:unicity}
		Let $\Conf\in\mc S_E$. For all $a\in E$ such that $a\in x$ for some $x\in\Conf$, there is a unique $p\in\pr(\Conf)$ that introduces $a$.
	\end{proposition}
	\begin{proof}
		This is a consequence of $\Conf$ being closed under intersection: by contradiction, suppose there are two prime elements $p$ and $q$ that introduce $a$. We would have $a\in p\cap q$ and $p\cap q\subseteq p$, which would contradict $a\not\in\pred_\Conf(p)$ (since $p\bs\pred_\Conf(p)=\set{a}$).
	\end{proof}
	
	\begin{corollary}[to Proposition~\ref{prop:unicity}]
		Let $\Conf\in\mc S_E$. For all finite $x\in\Conf$ and $p\in\pr(\Conf)$, there is a unique $q\in\pr(x\grpAct\Conf)$ such that $\der\Conf(p)=\der{x\grpAct\Conf}(q)$.
	\end{corollary}
	\begin{proof}
		Using the fact that symmetric residuation does not introduce nor remove events 
		and preserves stability (\autoref{thm:stable}).
	\end{proof}
	
	Thanks to the above corollary, we can associate prime elements of $\Conf$ to prime elements of the symmetric residual of $\Conf$ after a finite configuration:
	
	\begin{definition}
		\label{def:res-map}
		Let $\Conf\in\mc S_E$. For all finite $x\in\Conf$, define the \emph{(symmetric) residuation map} $\sigma_x:\pr(\Conf)\to\pr(x\grpAct\Conf)$ as $\sigma_x(p)=q$ if and only if $\der\Conf(p)=\der{x\grpAct\Conf}(q)$.
	\end{definition}
	
	\begin{example}\label{ex:rec1}
		Dashed arrows denote the residuation map $\sigma_{\{b, c\}}: \pr(\Conf) \to \pr(\{b, c\} \grpAct \Conf)$.
		
		
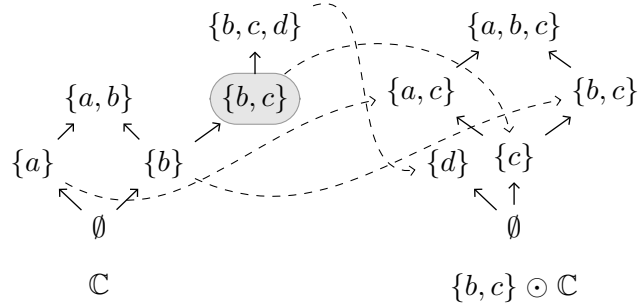
\noindent\centering
\begin{tikzpicture}
	\node (empt1) {\(\emptyset\)};
	\recCS
	
	\node (empt2) [right = 5cm of empt1] {\(\emptyset\)};
	\recsCS
	
	\draw[dmap] (a1) to[out=330, in=190] (ac2);
	\draw[dmap] (b1) to[out=335, in=190] (bc2);
	\draw[dmap] (bc1) to [out=40, in=100] (c2);
	\draw[dmap] (bcd1) to[out=20, in=190] (d2);
	
	\begin{scope}[on background layer]
		\node [fill=fillcolor, fit=(bc1), rounded corners=.35cm, inner sep=.3pt, draw=fillborder] {};
	\end{scope}	
\end{tikzpicture}
	\end{example}
	
	\begin{definition}
		We write  $x\ort\po y$ when two elements \(x, y \in \po\) are \emph{orthogonal}, \ie compatible and incomparable: if $x\comp_\po y$ and \(x \nsubseteq y\), \(y \nsubseteq x\). 
	\end{definition}
	
	We can 
	now 
	characterize algebraically the effect of symmetric residuation on prime elements.
	
	\begin{theorem}[Effects of residuation on causal structure]\label{thm:effect}
		Let $\Conf\in\mc S_E$. For all finite $x\in\Conf$, let $\sigma_x:\pr(\Conf)\to \pr(x\grpAct\Conf)$ be the residuation map of Definition~\ref{def:res-map}. 
		For all distinct $p,q\in\pr(\Conf)$, the following 
		 holds:
		\begin{align}
	\text{If $p\ort\Conf q$ then} && \sigma_x(p) & \ort{x\grpAct\Conf} \sigma_x(q) \label{eq:preserve-ort}\\
	\text{If $p\subseteq q$ then} && q\subseteq x &\implies \sigma_x(q)\subseteq \sigma_x(p) \label{eq:flip-cause}\\
	\text{If $p\subseteq q$ then} && p\not\subseteq x\wedge q\not\subseteq x &\implies \sigma_x(p)\subseteq\sigma_x(q) \label{eq:preserve-cause} \\
	\text{If $p\subseteq q$ then} && p\subseteq x \wedge q\not\subseteq x&\implies  \sigma_x(p) \ncomp_{x\grpAct\Conf} \sigma_x(q) \label{eq:cause-to-conflict}\\
	\text{If $p  \ncomp_{\Conf} q$ then} &&
	p\subseteq x &\implies \sigma_x(p)\subseteq\sigma_x(q) \label{eq:conflict-to-cause}\\
	\text{If $p \ncomp_{\Conf} q$ then} && p\not\subseteq x\wedge q\not\subseteq x& \implies \sigma_x(p) \ncomp_{x\grpAct\Conf} \sigma_x(q) \label{eq:preserve-conflict}
\end{align}
		Notice that for all $p',q'\in\pr(\Conf)$ exactly one of the above premise holds with either $p=p'$ and $q=q'$ or $p=q'$ and $q=p'$.
	\end{theorem}
	
	\begin{example}
		We illustrate \autoref{eq:conflict-to-cause} below, taking \(p = \{a\}\), \(q = \{b, d\}\) and \(x = \{a, c\}\), and representing the action of \(\sigma_{\set{a, c}}\) on \(\set{a}\) and \(\set{b, d}\) with dashed lines:	
		
		\noindent

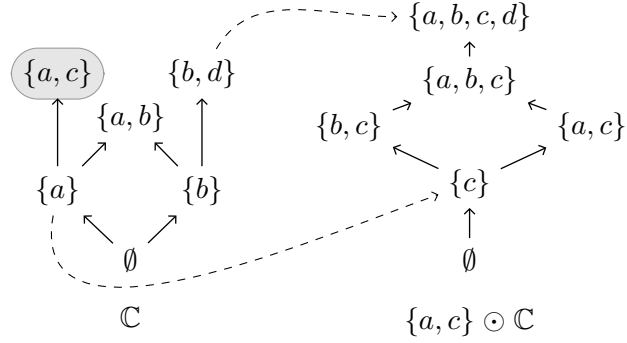
{
	
\centering

\begin{tikzpicture}
	\node (empt1) {\(\emptyset\)};
	\node (a1) [above left = .3cm and .3cm of empt1] {\(\{a\}\)};
	\node (b1) [above right  = .3cm and .3cm of empt1] {\(\{b\}\)};
	\node (c1) [above  = .5cm of empt1] {};
	\node (ac1) [above = .9cm of a1] {\(\{a, c\}\)};
	\node (bd1) [above = .9cm of b1] {\(\{b, d\}\)};
	\node (abc1) [above = .5cm of c1] {\(\{a, b\}\)};
	\node (caption1) [below = .2cm of empt1] {\(\Conf\)};
	\draw [su] (empt1) -- (a1);
	\draw [su] (empt1) -- (b1);
	\draw [su] (a1) -- (abc1);
	\draw [su] (a1) -- (ac1);
	\draw [su] (b1) -- (abc1);
	\draw [su] (b1) -- (bd1);
	\begin{scope}[on background layer]
		\node [fill=fillcolor, fit=(ac1), rounded corners=.35cm, inner sep=.3pt, draw=fillborder] {};
	\end{scope}	

	\node (empt2) [right = 4cm of empt1] {\(\emptyset\)};
	\node (c2) [above = .4cm of empt2] {\(\{c\}\)};
	\node (bc2) [above  left = .1cm and .6cm of c2] {\(\{b, c\}\)};
	\node (ac2) [above right = .1cm and .6cm of c2] {\(\{a, c\}\)};
	\node (abc2) [above = .7cm of c2] {\(\{a, b, c\}\)};
	\node (abcd2) [above = .2cm of abc2] {\(\{a, b, c, d\}\)};
	\node (caption2) [below = .2cm of empt2] {\(\{a, c\} \grpAct \Conf\)};
	\draw [su] (empt2) -- (c2);
	\draw [su] (c2) -- (bc2);
	\draw [su] (c2) -- (ac2);
	\draw [su] (bc2) -- (abc2);
	\draw [su] (ac2) -- (abc2);
	\draw [su] (abc2) -- (abcd2);
	
	\draw[dmap] (a1) to[out=260, in=200] (c2);
	\draw[dmap] (bd1) to[out=60, in=180] (abcd2);
	
\end{tikzpicture}

}
		
		We have \(\set{a} \not\comp_\Conf \set{b, d}\), \(\set{a} \subseteq \set{a, c}\) and indeed, \(\sigma_x(\set{a}) = \{c\} \subseteq \sigma_{\set{a, c}}(\set{b, d}) = \{a, b, c, d\}\).
	\end{example}
	
	\begin{proofwq}
		In the following we use the notation $x:a$ for $a\in x$, $x:\bar b$ for $b\not\in x$ and use concatenations like $x:a\bar b$, for conjunction of these properties. Suppose $p$ introduces $a$ and $q$ introduces $b$ in $\Conf$. We verify \autoref{eq:preserve-ort} -- \autoref{eq:preserve-conflict} by case analysis.
		\begin{itemize}
			\item[\autoref{eq:preserve-ort}] We must check that $p\ort\Conf q$ implies $\sigma_x(p)\ort{x\grpAct\Conf}\sigma_x(q)$. 
			Since $p\ort\Conf q$ by hypothesis, we have $p:a\bar b$ and $q:a \bar b$, otherwise by Proposition~\ref{prop:prime-prop} we would have either $p\subseteq q$ or $q\subseteq p$ and they would not be orthogonal. We have three sub-cases to consider.
			\begin{itemize}
				\item[\autoref{eq:preserve-ort}.1] Suppose $x:ab$. By Proposition~\ref{prop:prime-prop} we have $p\subseteq x$ and $q\subseteq x$. It entails that $p\Del x=x\bs p$ and $q\Del x=x\bs q$ are configurations of $x\grpAct\Conf$. Since $p:a\bar b$, and $x:ab$, we have $x\bs p:\bar ab$. By the same reasoning we have $x\bs q:a\bar b$. Now, $x\bs p$ and $x\bs q$ are below $x$ and incomparable. Since $x\bs p:b$ we have $\sigma_x(q)\subseteq x\bs p$ (by Proposition~\ref{prop:prime-prop}) and $\sigma_x(p)\subseteq x\bs q$ by the same argument. It follows that $\sigma_x(p)\ort{x\grpAct\Conf}\sigma_x(q)$.
				
				\item[\autoref{eq:preserve-ort}.2] Suppose $x:\bar ab$ (the argument for $x:a\bar b$ is symmetrical). We have $p\cap x\in\Conf$ (by stability of $\Conf$) and $p\cup q\in\Conf$ (since $p\ort\Conf q$). It entails that $x\Del (p\cap x) = x\bs p\in x\grpAct\Conf$ and $x\Del(p\cup q)\in x\grpAct\Conf$. Now 
				$x:\bar ab$ and $p:a\bar b$ implies $x\bs p:\bar ab$. Also, $p\cup q:ab$ implies $x\Del(p\cup q):a\bar b$. It entails that $\sigma_x(q)\subseteq x\bs p$ and $\sigma_x(p)\subseteq x\Del(p\cup q)$ (by Proposition~\ref{prop:prime-prop}). Additionally $x:\bar ab$ and $p:a\bar b$ implies $x\Del p:ab$ and is a configuration of $x\grpAct\Conf$. It follows that $\sigma_x(p)\subseteq x\Del p$ and $\sigma_x(q)\subseteq x\Del p$, so they are compatible. We prove they are incomparable by contradiction. Suppose $\sigma_x(p)\subseteq \sigma_x(q)$. This would entail $\sigma_x(p)\subseteq x\bs p$ which would contradict $x\bs p:\bar a$ (by Proposition~\ref{prop:prime-prop}). Still by contraction suppose $\sigma_x(q)\subseteq \sigma_x(p)$. This would entail $\sigma_x(q)\subseteq x\Del(p\cup q)$ which would contradict $x\Del(p\cup q):\bar b$. We therefore have $\sigma_x(p)\ort{x\grpAct\Conf}\sigma_x(q)$. 
				
				\item[\autoref{eq:preserve-ort}.3] Suppose $x:\bar a\bar b$. We have $x\Del (p\cup q):ab$, $x\Del p:a\bar b$ and $x\Del q:\bar ab$ are all configurations of $x\grpAct\Conf$. As a consequence $\sigma_x(p)\subseteq x\Del p$, $\sigma_x(q)\subseteq x\Del q$, $\sigma_x(p)\subseteq x\Del (p\cup q)$ and $\sigma_x(q)\subseteq x\Del (p\cup q)$. Therefore $\sigma_x(p)$ and $\sigma_x(q)$ are compatible. We prove they are incomparable by contradiction. Suppose $\sigma_x(p)\subseteq \sigma_x(q)$. It would entail that $\sigma_x(p)\subseteq x\Del q$ which would contradict $x\Del q:\bar a$. Still by contradiction suppose $\sigma_x(q)\subseteq \sigma_x(p)$. This would imply that $\sigma_x(q)\subseteq x\Del p$ which would contradict $x\Del p:\bar b$. Therefore $\sigma_x(p)\ort{x\grpAct\Conf}\sigma_x(q)$.
			\end{itemize}
			\item[\autoref{eq:flip-cause}] We verify $p\subseteq q$ and $q\subseteq x$ implies $\sigma_x(q)\subseteq \sigma_x(p)$. Combining hypothesis and  Proposition~\ref{prop:prime-prop}, we have $p:a\bar b$, $q:ab$ and $x:ab$.
			By contradiction, we show that $\sigma_x(q)\subseteq \sigma_x(p)$. Suppose not, in which case we have $\sigma_x(p):a\bar b$. So we have that $x\Del\sigma_x(p):\bar ab$ is a configuration of $\Conf$ in which $b$ appears without $a$. This violates the fact that $q:ab$ introduces $b$ in $\Conf$ and contains $a$. Therefore $\sigma_x(q)\subseteq \sigma_x(p)$.
			
			\item[\autoref{eq:preserve-cause}]  We verify $p\subseteq q$ and $p\not\subseteq x\wedge q\not\subseteq x$ implies $\sigma_x(p)\subseteq\sigma_x(q)$. Again by contradiction, suppose $\sigma_x(p)\not\subseteq\sigma_x(q)$. We would have $\sigma_x(q):\bar ab$ and there would be a configuration $x\Del\sigma_x(q):\bar ab$ in $\Conf$, since by hypothesis $x:\bar a\bar b$. This would contradict the fact that any configuration containing $b$ must be above $q:ab$ and hence should also contain $a$. Therefore $\sigma_x(p)\subseteq \sigma_x(q)$. 
			
			\item[\autoref{eq:cause-to-conflict}] We verify $p\subseteq q$, $p\subseteq x$ and $q\not\subseteq x$ implies $\sigma_x(p)\not\comp_{x\grpAct\Conf} \sigma_x(q)$. By contradiction, suppose there is a configuration $y\in x\grpAct\Conf$ such that $\sigma_x(p)\subseteq y$ and $\sigma_x(q)\subseteq y$. We would have $y:ab$ and there would be a configuration $x\Del y\in \Conf$. Now by hypothesis we have $x:a\bar b$, therefore $x\Del y:\bar ab$. This would contradict the hypothesis $q:ab$ and therefore $\sigma_x(p)\not\comp_{x\grpAct\Conf} \sigma_x(q)$.
			
			\item[\autoref{eq:conflict-to-cause}] We verify $p\not\comp_\Conf q$ and $p\subseteq x$ implies $\sigma_x(p)\subseteq\sigma_x(q)$. By contradiction, suppose that $\sigma_x(p)\not\subseteq\sigma_x(q)$. We would have $\sigma_x(q):\bar ab$ and $x\Del\sigma_x(q):ab$ would be a configuration of $\Conf$. This would contradict the hypothesis $p\not\comp_\Conf q$ and therefore $\sigma_x(p)\subseteq\sigma_x(q)$. 
			
			\item[\autoref{eq:preserve-conflict}] We verify $p\not\comp_\Conf q$ and $p\not\subseteq x$ and $q\not\subseteq x$ implies $\sigma_x(p)\mkern-1mu \ncomp_{x\grpAct\Conf}\mkern1mu\sigma_x(q)$. By contradiction, suppose there is a configuration $y\in x\grpAct\Conf$ such that $\sigma_x(p)\subseteq y$ and $\sigma_x(q)\subseteq y$. We would have $y:ab$ and $x\Del y:ab$ would be a configuration of $\Conf$. This contradicts the fact that $p$ and $q$ are incompatible. Therefore, $\sigma_x(p)\mkern-1mu \ncomp_{x\grpAct\Conf}\mkern1mu\sigma_x(q)$. \qed \end{itemize}
	\end{proofwq}
	
	\section{A switch operation on prime event structures}\label{sec:switch}
	
	Now that the effects of symmetric residuation are fully characterized, we proceed with the definition of symmetric residuation on prime event structures.
	
	\begin{definition}
\label{def:switch}
		Let $\ES=(E,<,\cf)$ be a prime event structure. For all finite configuration $X\in\mathit{Conf}(\ES)$ we define $X\grpAct\ES\coloneqq(E,<',\cf')$ as:
		\begin{align*}
			a<'b & \text{ if and only if } \begin{dcases*}
				a<b \text{ and }\set{a,b}\cap X=\emptyset\\
				b<a \text{ and }\set{a,b}\subseteq X\\
				a\cf b \text{ and }a\in X
			\end{dcases*} &&&&
			a\cf'b & \text{ if and only if }
			\begin{dcases*}
				a\cf b \text{ and } \set{a,b}\cap x=\emptyset\\
				a<b, a\in X\text{ and } b\not\in X
			\end{dcases*}
		\end{align*}
	where clauses on the right of the curly brackets are taken disjunctively.
	\end{definition}
	
	Note that the definition above does not imply that $X\grpAct\ES$ is a prime event structure, as the $<'$ and $\cf'$ might not satisfy the requirements listed in Definition~\ref{def:PES}. Theorem \autoref{thm:switch} will take care of both by ensuring that $X\grpAct\ES$ is indeed a group action on prime event structure, and that it properly reflects the effect of symmetric residuation of the underlying configuration structure.
	
	\begin{example}\label{ex:rec2}
		A prime event structure \(\ES\) with its configuration \(X = \{\{b\}, \{b, c\}\}\) grayed out, and \(X \grpAct \ES\):
		
		
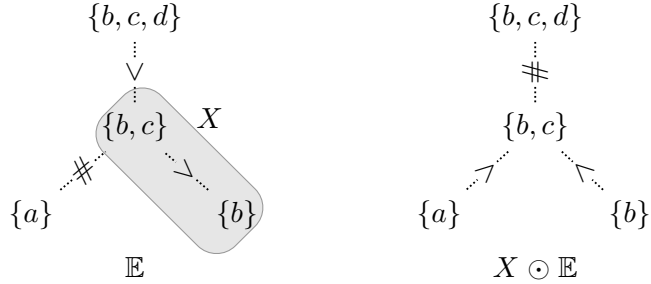
\noindent{\centering

\begin{tikzpicture}
	\node (marq1) {}; 
	\recES
	\node (marq2) [right = 5cm of marq1] {}; 
	\recsES
\end{tikzpicture}

}
	\end{example}
	
	Switching an event structure is relative to a finite configuration $X$ and is reminiscent of the Seidel switch, a  well known graph isomorphism which consists in fixing a set of vertices $X$ and flipping edges and non-edges between $X$ and its complement in the graph (the other edges remaining invariant) \cite{Seidel73}. We discuss potential interesting connections in the conclusion of this paper.
	
	For all configuration structures $\Conf$, recall that $\mathsf{E}(\Conf)$ is Winskel's construction, mapping sets of configurations to a prime event structure, via the prime elements (see \autoref{thm:ESofConf}). The following theorem, where we borrow from Definition~\ref{def:prime} that for all $x\in\Conf$, $\dec{x}_\Conf\coloneqq {\down}_\Conf\set{x}\cap\pr(\Conf)$, characterizes symmetric residual as a switch operation on prime event structures. 
	
	\begin{theorem}[Adequacy]\label{thm:switch}
		Let $\Conf\in\mc S_E$ be a stable configuration structure.
		For all finite $x\in\Conf$, the prime event structure $\mathsf{E}(x\grpAct\Conf)$ is isomorphic to $\dec{x}_\Conf\grpAct\mathsf{E}(\Conf)$.
	\end{theorem}
	
	\begin{proofwq}
		We show that the isomorphism of event structures is given by the bijection $\sigma_x$ of Definition~\ref{def:res-map}. Consider
		\begin{alignat*}{3}
			& \ES_0 && \coloneqq \mathsf{E}(\Conf)         && =(\pr(\Conf),<_0,\cf_0)\\
			& \ES_1 && \coloneqq \dec{x}_\Conf\grpAct\ES_0   && = (\pr(\Conf),<_1,\cf_1) \\
			& \ES_2 && \coloneqq \mathsf{E}(x\grpAct\Conf) &&=(\pr(x\grpAct\Conf),<_2,\cf_2)
		\end{alignat*}
		
		We overload notations and write $a\ort{} b$ for two events in $(E,<,\cf)$ whenever $a\not< b$, $b\not< a$ and $a\not\cf b$. Since $\sigma_x^{-1}$ is the residuation map of $x\grpAct (x\grpAct\Conf)=\Conf$ the proof is symmetric and we just need to show $p\mc R_1q$ implies $\sigma_x(p)\mc R_2\sigma_x(q)$, for $\mc R_i\in\set{\ort i, <_i, \cf_i}$:
		\begin{description}
			\item[$p\ort 1q$ implies $\sigma_x(p)\ort 2\sigma_x(q)$.]
			We have $p\ort 1q$ iff $p\ort 0q$, by definition of the switch operation. Now $p\ort0q$ implies $\sigma_x(p)\ort0\sigma_x(q)$ by \autoref{thm:effect}~\autoref{eq:preserve-ort}. 
			\item[$p<_1 q$ implies $\sigma_x(p)<_2 \sigma_x(q)$.] There are three sub-cases according to Definition~\ref{def:switch}:
			\begin{itemize}
				\item $p <_0 q$ and $\set{p,q}\cap \dec{x}_\Conf=\emptyset$. We have $p<_0q$ implies $p\subseteq q$ by \autoref{thm:ESofConf}. 
				Since $\Conf$ is prime algebraic, $\Sup_\Conf\dec{x}_\Conf=x$ and we have $p\not\subseteq x$ and $q\not\subseteq x$. By \autoref{thm:effect}~\autoref{eq:preserve-cause}, this implies $\sigma_x(p)\subseteq \sigma_x(q)$ and hence $\sigma_x(p)<_2 \sigma_x(q)$.	
				
				\item $q <_0 p$ and $\set{p,q}\subseteq \dec{x}_\Conf$. We have $q<_0p$ implies $q\subseteq p$ by \autoref{thm:ESofConf}. Since $\Conf$ is prime algebraic, $\Sup_\Conf\dec{x}_\Conf=x$ and we have $p\subseteq x$ and $q\subseteq x$. By \autoref{thm:effect}~\autoref{eq:flip-cause}, we have $\sigma_x(p)\subseteq\sigma_x(q)$ which entails $\sigma_x(p)<_2\sigma_x(q)$ by definition of $<_2$ (see \autoref{thm:ESofConf}).
				
				\item $p\cf_0 q$ and $p\in \dec{x}_\Conf$. Using a similar argument as above, we deduce $p\subseteq x$. Additionally since $p\not\comp_{\ES_0}q$, we can deduce $q\not\subseteq x$, otherwise $x$ would be a bound for $p$ and $q$. We can apply \autoref{thm:effect}~\autoref{eq:conflict-to-cause} to obtain $\sigma_x(p)\subseteq \sigma_x(q)$, and as a consequence $\sigma_x(p)<_2 \sigma_x(q)$.
			\end{itemize}
			\item[$p\cf_1 q$ implies $\sigma_x(p)\cf_2 \sigma_x(q)$.] We have two sub-cases according to Definition~\ref{def:switch}:
			\begin{itemize}
				\item $p\cf_0q$ and $\set{p,q}\cap \dec{x}_\Conf=\emptyset$. We have $p\not\comp_\Conf q$, $p\not\subseteq x$ and $q\not\subseteq x$. Applying \autoref{thm:effect}~\autoref{eq:preserve-conflict}, we have $\sigma_x(p)\not\comp_{x\grpAct\Conf}\sigma_x(q)$ and thus $\sigma_x(p)\cf_2\sigma_x(q)$.
				
				\item $p<_0 q$, $p\in\dec{x}_\Conf$ and $q\not\in\dec{x}_\Conf$. We have $p\subseteq q$, $p\subseteq x$ and $q\not\subseteq x$. By \autoref{thm:effect}~\autoref{eq:cause-to-conflict}, we get $\sigma_x(p)\not\comp_{x\grpAct\Conf}\sigma_x(q)$. Therefore $\sigma_x(p)\cf_2\sigma_x(q)$.
				\qed 
			\end{itemize}
		\end{description}
	\end{proofwq}
	
	\begin{example}
		We show in \autoref{fig:sum-example} an illustration of \autoref{thm:switch}, connecting Examples~\ref{ex:rec1} and \ref{ex:rec2}.
	\end{example}
	
	\begin{figure}
		\noindent{
	\centering

\begin{tikzpicture}	
	\node (empt1) {\(\emptyset\)};
	\recCS
	
	\node (empt2) [right = 5cm of empt1] {\(\emptyset\)};
	\recsCS
	
	\begin{scope}[on background layer]
		\node [fill=fillcolor, fit=(bc1), rounded corners=.35cm, inner sep=.3pt, draw=fillborder] {};
	\end{scope}	

	\node [below = 4.5cm of empt1] (marq1) {}; 
	\recES
	\node (marq2) at (marq1 -| empt2) {}; 
	\recsES

	\node (marq3) at ($(empt1)!0.5!(marq2)+(7,0)$) {}; 
	\node (a3) [left = .7cm of marq3] {$\set{a, c}$};
	\node (b3) [right = .7cm of marq3] {\(\{b, c\}\)};
	\node (bc3) [above = .7cm of marq3] {\(\{c\}\)};
	\node (bcd3) [above = .8cm of bc3] {\(\{d\}\)};
	\node (ES3) [below = .3cm of marq3] {\(\ES(\set{b, c} \grpAct \Conf)\)};
	\draw [cf] (bc3) -- (bcd3);
	\draw [ord] (bc3) -- (b3);
	\draw [ord] (bc3) -- (a3);

	\draw[smap] ($(empt1) + (1, 0)$) to [bend right] node[mapl]{\(\grpAct_{\{b, c\}}\)} ($(empt2) + (-1, 0)$);
	\draw[smap] ($(marq1) + (1, .8)$) to [bend left] node[mapl]{\(\grpAct_{X}\)} ($(marq2) + (-1, .8)$);
	\draw[umap, transform canvas={xshift=-0.6ex}] ($(CS1) + (0, -.3)$) to node[right = .3em]{\(\mathsf{C}\)} ($(marq1) + (0, 2.8)$);
	\draw[umap, transform canvas={xshift=0.6ex}] ($(marq1) + (0, 2.8)$) to node[left = .3em]{\(\mathsf{E}\)} ($(CS1) + (0, -.3)$);
	\draw[umap, transform canvas={yshift=-.8ex}] ($(CS2) + (1 , 2)$) to  node[below = 0em]{\(\mathsf{E}\)} ($(bc3) + (-1, 0)$);
	\draw[umap, transform canvas={yshift=.8ex}] ($(bc3) + (-1,0)$) to  node[above = 0em]{\(\mathsf{C}\)} ($(CS2) + (1, 2)$);	
	\draw[smap, transform canvas={xshift=1ex}] ($(ES3) + (0, -.5)$) to [bend left] node[mapl]{\(\sigma_{\set{b, c}}\)} ($(marq2) + (1, .6)$);
\end{tikzpicture}

}
		\caption{Illustrating the correspondence between symmetric residuation and switch operations, with \(X = \{\{b\}, \{b, c\}\}\).}\label{fig:sum-example}
	\end{figure}
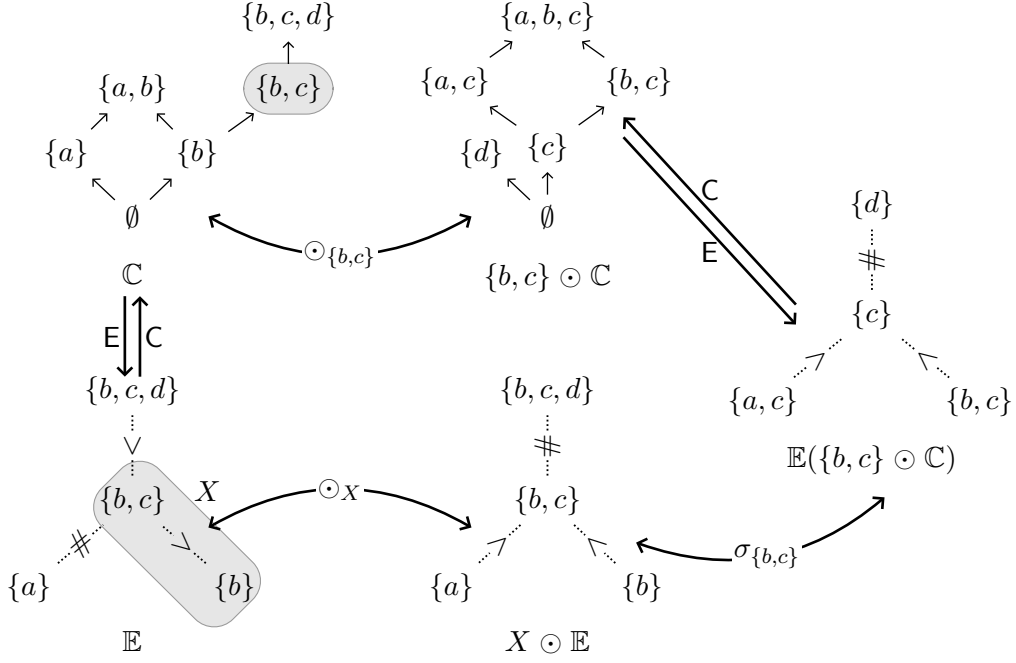
	
	\begin{corollary}[To \autoref{thm:switch}]
		Let $\Conf\in\mc S_E$ be a stable configuration structure. For all finite $x\in\Conf$, $\dec{x}_\Conf\grpAct\mathsf{E}(\Conf)$ is a prime event structure.
	\end{corollary}
	
	\section{A switch operation on polarized event structure}\label{sec:pointedes}
	
	The switch operation on a configuration $X$ reverses causality within $X$ and interchanges conflict and causality at its boundary (see Definition~\ref{def:switch}). Starting from pointed configuration structures, we derived a polarized event structure (Definition~\ref{def:polarized-event-structure}), where events are classified as either positive or negative based on their orientation relative to a given computational past. We now demonstrate that switching a polarized event structure on a configuration $X$ also reverses the signs of the events within $X$: consuming positive events results in negative ones, while backtracking on negative events restores their positive counterparts.
	
	\begin{definition}
\label{def:switch-polarized}
		Let 
		$\ES[\pi]$ be a polarized event structure. We define its \emph{switching} as:
		\(X\grpAct\ES[\pi] \coloneqq (X\grpAct\ES)[\pi']\)
		where for all $e\in E$, 
		\begin{align*}
			\pi'(e) \coloneqq \begin{dcases*}
				-\pi(e) & if \(e \in X\)\\
				\pi(e) & otherwise.
			\end{dcases*}
		\end{align*}
	\end{definition}
	
	Similarly to the switch operation on prime event structures, 
	the above operation
	does not guarantee to produce
	a polarized event structure, as the negative elements of $X\grpAct\ES[\pi]$ could fail to be downward closed for the causality relation, or exhibit conflict. The theorem below proves that it is not the case.
	
	\begin{theorem}[Polarity switch]\label{thm:pol-switch}
		Let $\Conf[y^\dagger]\in\mc S_E$ be a stable pointed configuration structure.
		For all finite $x\in\Conf$, the polarized event structure  $\dec{x}_\Conf\grpAct\mathsf{E}(\Conf[y^\dagger])$ is isomorphic to $\mathsf{E}(x\grpAct\Conf[y^\dagger])$.
	\end{theorem}
	\begin{proofwq}
		Consider: 
		\begin{align*}
			\ES_0[\pi_0] \coloneqq \mathsf{E}(\Conf[y^\dagger]) &&
			\ES_1[\pi_1] \coloneqq \dec{x}_\Conf\grpAct\ES_0[\pi_0] &&
			\ES_2[\pi_2] \coloneqq \mathsf{E}(x\grpAct\Conf[y^\dagger])
		\end{align*}
		To show that $\ES_1[\pi_1]$ is isomorphic to $\ES_2[\pi_2]$ we must \begin{enumerate*}[label=(\roman*)]
			\item exhibit an isomorphism between $\ES_1$ and $\ES_2$ and \label{item-proof-1}
			\item show that it is polarity preserving.\label{item-proof-2}
		\end{enumerate*}
		We already know from \autoref{thm:switch} that $\ES_1\sim_{\sigma_x}\ES_2$, so we just need to verify \ref{item-proof-2}, \ie $\pi_1(p)=\pi_2(\sigma_x(p))$ for all $p\in\pr(\Conf)$. According to Definition~\ref{def:switch-polarized},  we have two cases:
		\begin{itemize}
			\item If $p\in\dec{x}_\Conf$  then $\pi_1(p)=-\pi_0(p)$. By Proposition~\ref{prop:polarized-winksel}, we have two sub-cases to consider: 
			\begin{itemize}
				\item Suppose $p\subseteq y^\dagger$ and hence $\pi_0(p)<0$ and $\pi_1(p)>0$.  We must verify $\pi_2(\sigma_x(p))>0$. We have $p\subseteq x$ and $p\subseteq y^\dagger$ implies $a\not\in x\Del y^\dagger$ and therefore $\sigma_x(p)\not\subseteq x\Del y^\dagger$ (by Proposition~\ref{prop:prime-prop}). Hence, $\pi_2(\sigma_x(p))>0$. 
				\item Suppose  $p\not\subseteq y^\dagger$ and hence $\pi_0(p)>0$ and $\pi_1(p)<0$.  We must verify $\pi_2(\sigma_x(p))<0$. We have $p\subseteq x$ and $p\not\subseteq y^\dagger$ implies $a\in x\Del y^\dagger$ and therefore $\sigma_x(p)\subseteq x\Del y^\dagger$ (by Proposition~\ref{prop:prime-prop}). Hence, $\pi_2(\sigma_x(p))<0$. 
			\end{itemize}
			\item If $p\not\in\dec{x}_\Conf$ then $\pi_1(p)=\pi_0(p)$. We consider again the two sub-cases of Proposition~\ref{prop:polarized-winksel}:
			\begin{itemize}
				\item Suppose $p\subseteq y^\dagger$ and hence $\pi_0(p)<0$ and $\pi_1(p)<0$.  We must verify $\pi_2(\sigma_x(p))<0$. We have $p\not\subseteq x$ and $p\subseteq y^\dagger$ implies $a\in x\Del y^\dagger$ and therefore $\sigma_x(p)\subseteq x\Del y^\dagger$ (by Proposition~\ref{prop:prime-prop}). Hence, $\pi_2(\sigma_x(p))<0$. 
				\item Suppose  $p\not\subseteq y^\dagger$ and hence $\pi_0(p)>0$ and $\pi_1(p)>0$.  We must verify $\pi_2(\sigma_x(p))>0$. We have $p\not\subseteq x$ and $p\not\subseteq y^\dagger$ implies $a\not\in x\Del y^\dagger$ and therefore $\sigma_x(p)\not\subseteq x\Del y^\dagger$ (by Proposition~\ref{prop:prime-prop}). Hence, $\pi_2(\sigma_x(p))>0$.  \qed
			\end{itemize}
		\end{itemize}
	\end{proofwq}
	
	\begin{corollary}[To \autoref{thm:pol-switch}]
		Let $\Conf\in\mc S_E$ be a stable configuration structure and $y^\dagger\in\Conf$ be a distinguished finite configuration. For all finite $x\in\Conf$, $\dec{x}_\Conf\grpAct\mathsf{E}(\Conf[y^\dagger])$ is a polarized event structure.
	\end{corollary}
	
	\section{Discussion}\label{sec:ccl}
	Reversibility is a fundamental property in concurrent systems with applications ranging from debugging to quantum computing, yet its integration with traditional operational semantics remains challenging
	
	In the past two decades, significant research has focused on the design and formalization of reversible concurrent programs~\cite{DanKri04,PhiUli06,Lan_etal11,DomKriVar13,Lan_etal18}. This research typically builds upon a classical, non-reversible operational semantics represented as a transition system, say $\sem{P}\coloneqq (P,\mc P,\to)$, and turns it into a reversible one, say $\sem{P}_\mathsf{rev}\coloneqq(P,\mc R,\to_\mathsf{rev})$, where $\mc P$ and $\mc R$ are the program states that are reachable from the initial state $P$, under the transition relations. Importantly, verifying that $\sem{\cdot}_\mathsf{rev}$ is sound involves proving
	\begin{enumerate*}[label=(\roman*)]
		\item that $\mc P$ and $\mc R$ are essentially the same states, after removing additional decorations or encodings that are required to make the transition system reversible, and\label{cond1}%
		\item that backtracking is causally consistent, \ie that it is not possible to backtrack on a transition without backtracking first on its consequences. \label{cond2}
	\end{enumerate*}
	Condition \ref{cond1} ensures that the reversible semantics does not create new program states that were not reachable in the forward only semantics. Verification techniques to ensure \ref{cond2} have been narrowed down to requiring that $\to_\mathsf{rev}$ satisfies the fundamental \emph{causal consistency} property, which states that all co-initial and co-final sequences of transitions, \ie of the form $P\to_\mathsf{rev}^{a_1}\dots\to_\mathsf{rev}^{a_n}Q$ and $P\to_\mathsf{rev}^{b_1}\dots\to_\mathsf{rev}^{b_q}Q$, can be made \emph{equal} by permuting concurrent steps and cancelling out consecutive inverse ones \cite{DanKri04,Lan_etal14,LPU24,Aub24}. Importantly, proofs of causal consistency are syntax-intensive proofs which require characterizing causality, conflict and independence of transitions for each formalism whose reversible semantics is being established. We let the reader refer to a more axiomatic approach to reversibility~\cite{LPU24} for a principled way of obtaining this result.
	
	The work presented in this paper proposes
approaching reversible computations in a syntax-free manner, assuming that a concurrent program is denoted by the set of its observable configurations. This makes practically no assumption on the nature of the underlying process, as configuration structures are just a set of observations, ordered by inclusion. We then showed that the reversible interpretation of a process can be made apparent by changing the traditional interpretation of the residuation of a configuration structure (Definition~\ref{def:residuation}) \emph{after} some events have occurred: instead of erasing those events, we showed that it becomes possible to observe inverse events using a \emph{symmetric residuation operation} (Definition~\ref{def:sym-residuation}) that is shown to be a group action of configuration structures on themselves (Proposition~\ref{prop:gp-act}).
	
	In a second step, we strengthened the hypothesis on the causal structure of the underlying concurrent process by assuming it can be denoted by a prime event structure. As already well-established in the literature, this corresponds to requiring some stability axioms on its configuration structures (Definition~\ref{def:stable}). Note that stable denotational semantics for concurrent formalisms does not come for free, as multiple formalisms are known to naturally yield unstability \cite{BorSan98,Cri_etal16,Bal_etal17}. 
	At this point we showed that stable configuration structures are preserved  by symmetric residuation (\autoref{thm:stable}), which guarantees that the symmetric residual of a configuration structure is also a prime event structure:
	
	\begin{center}	
		\begin{tikzpicture}
			\node (C) {\(\Conf\)};
			\node(xC)[right = 2cm of C]{$x\grpAct\Conf$};
			\node(E)[below of = C]{$\ES$};
			\node(E')[below of = xC]{$\ES'$};
			\draw [su] (C) -- node[above] {\small\emph{Sym.\ res.}}(xC);
			\draw [su, transform canvas={xshift=2pt}] (C) -- (E);
			\draw [su, transform canvas={xshift=-2pt}] (E) -- (C);
			\draw [su, transform canvas={xshift=2pt}] (xC) -- (E');
			\draw [su, transform canvas={xshift=-2pt}] (E') -- (xC);
			\draw [su, dotted] (E) -- (E');
		\end{tikzpicture}
	\end{center}
	We then showed that one can instantiate the dotted arrow of the above diagram using a switch operation on prime event structures (\autoref{thm:switch}) which turns causality upside-down amongst events of a configuration, and flips conflict and causality relation at its boundary.   
	
	We believe this work lays a foundation for developing a general theory of reversibility in concurrent systems, while also opening up new avenues, some of which we discuss now.	
	
	First, notice that rewriting a configuration structure $\Conf$ into $x\grpAct\Conf$ is a global operation which transforms every configuration of $\Conf$ using symmetric set difference, whereas the switch operation is purely local
	, since the causality and conflict relations are only modified on pairs of events that appear in $x$. This indicates a promising technique to verify causal consistency of a reversible operational semantics: reusing the informal notations introduced above, one first needs to prove that $\sem{P}_\mathsf{rev}$ exhibit a stable causal semantics (its configuration structures are stable), which is a necessary and sufficient condition to establish that prime event structures can be used to denote computations of $P$. It then suffices to demonstrate that the operational semantics is \enquote{switch-like}, \ie a computation event $a:P\to_\mathsf{rev}Q$ produces the effect of a switch on the causal and conflict structures of $P$, which is a predicate on immediate predecessors and successors of $a$, and on events that are in conflict with $a$ as well (see Definition~\ref{def:switch}). Requiring the reversible operational semantics to be stable and switch-like would then be an alternative characterization of being consistent.
	
	Another intriguing perspective comes from the operation of symmetric residuation, that could be studied as a symmetry on partial order. The definition of the operation $x\grpAct\Conf$ requires a set structure and is therefore not directly a transformation of partial orders. However we have shown that the operation does act as a switch on prime event structures, which can be interpreted as particular kinds of 2 sorted directed graph. With this perspective, the switch introduced in this paper becomes an operation close to a Seidel switch \cite{Seidel73}, originally defined as a symmetry on graphs that consists in swapping edges and non edges that are incident on nodes in a set $X$, which is the parameter of the switch. We believe the actual relation between our event structure switch and the Seidel switch is worth investigating as Seidel switches on directed graphs is known to produce symmetries of partial orders \cite{Pac_etal13} and have potential interpretations in quantum computing \cite{Sup_etal16}, which enjoys reversibility of operations prior to measurements.
	
	\bibliographystyle{./entics}
	\bibliography{bib/bibexport}

\begin{thebibliography}{10}
\providecommand{\url}[1]{\texttt{#1}}
\providecommand{\urlprefix}{ }
\providecommand{\eprint}[2][]{\url{#2}}

\bibitem{Abr05}
Abramsky, S., \emph{A structural approach to reversible computation},
  Theoretical Computer Science \textbf{347}, pages 441--464 (2005).
\newline\urlprefix\url{https://doi.org/10.1016/j.tcs.2005.07.002}

\bibitem{Aub24}
Aubert, C., \emph{The correctness of concurrencies in (reversible) concurrent
  calculi}, Journal of Logical and Algebraic Methods in Programming
  \textbf{136}, page 100924 (2024), ISSN 2352-2208.
\newline\urlprefix\url{https://doi.org/10.1016/j.jlamp.2023.100924}

\bibitem{Aubert2020b}
Aubert, C. and I.~Cristescu, \emph{How reversibility can solve traditional
  questions: The example of hereditary history-preserving bisimulation}, in:
  I.~Konnov and L.~Kov\'{a}cs, editors, \emph{31st International Conference on
  Concurrency Theory, {CONCUR} 2020, September 1--4, 2020, Vienna, Austria},
  volume 2017 of \emph{Leibniz International Proceedings in Informatics}, pages
  13:1--13:24, Schloss Dagstuhl--Leibniz-Zentrum f{\"u}r Informatik (2020).
\newline\urlprefix\url{https://doi.org/10.4230/LIPIcs.CONCUR.2020.7}

\bibitem{Aubert2025}
Aubert, C., I.~Phillips and I.~Ulidowski, \emph{Bisimulations and
  reversibility}, in: \emph{Components operationally: Reversibility and system
  engineering}, volume 16065 of \emph{LNCS}, Springer International Publishing
  (2025).
\newline\urlprefix\url{https://doi.org/10.1007/978-3-031-99717-4_3}

\bibitem{Bal_etal17}
Baldan, P., A.~Corradini and F.~Gadducci, \emph{Domains and event structures
  for fusions}, in: \emph{32nd Annual {ACM/IEEE} Symposium on Logic in Computer
  Science, {LICS} 2017, Reykjavik, Iceland, June 20-23, 2017}, pages 1--12,
  {IEEE} Computer Society (2017).
\newline\urlprefix\url{https://doi.org/10.1109/LICS.2017.8005135}

\bibitem{Bar_etal18}
Barylska, K., M.~Koutny, {\L}.~Mikulski and M.~Pi{\k a}tkowski,
  \emph{Reversible computation vs. reversibility in petri nets}, Science of
  Computer Programming \textbf{151}, pages 48--60 (2018), ISSN 0167-6423.
  Special issue of the 8th Conference on Reversible Computation (RC2016).
\newline\urlprefix\url{https://doi.org/10.1016/j.scico.2017.10.008}

\bibitem{BorSan98}
Boreale, M. and D.~Sangiorgi, \emph{A fully abstract semantics for causality in
  the {$\pi$}-calculus}, Acta Inf. \textbf{35}, pages 353--400 (1998).

\bibitem{BouCas87}
Boudol, G. and I.~Castellani, \emph{On the semantics of concurrency: Partial
  orders and transition systems}, in: H.~Ehrig, R.~Kowalski, G.~Levi and
  U.~Montanari, editors, \emph{{TAPSOFT'87}}, volume 249 of \emph{Lecture Notes
  in Computer Science}, pages 123--137, Springer, Berlin, Heidelberg (1987),
  ISBN 3-540-17660-8.
\newline\urlprefix\url{https://doi.org/10.1007/3-540-17660-8_52}

\bibitem{BouCas89}
Boudol, G. and I.~Castellani, \emph{Permutation of transitions: An event
  structure semantics for {CCS} and {SCCS}}, in: J.~W. de~Bakker, W.~P.
  de~Roever and G.~Rozenberg, editors, \emph{Linear Time, Branching Time and
  Partial Order in Logics and Models for Concurrency, School/Workshop,
  Noordwijkerhout, The Netherlands, May 30 - June 3, 1988, Proceedings}, volume
  354 of \emph{Lecture Notes in Computer Science}, pages 411--427, Springer
  (1988), ISBN 3-540-51080-X.
\newline\urlprefix\url{https://doi.org/10.1007/BFb0013028}

\bibitem{BuhTroVit01}
Buhrman, H., J.~Tromp and P.~M.~B. Vit{\'{a}}nyi, \emph{Time and space bounds
  for reversible simulation}, in: F.~Orejas, P.~G. Spirakis and J.~van Leeuwen,
  editors, \emph{Automata, Languages and Programming, 28th International
  Colloquium, {ICALP} 2001, Crete, Greece, July 8-12, 2001, Proceedings},
  volume 2076 of \emph{Lecture Notes in Computer Science}, pages 1017--1027,
  Springer (2001), ISBN 3-540-42287-0.
\newline\urlprefix\url{https://doi.org/10.1007/3-540-48224-5_82}

\bibitem{DomKriVar13}
Cristescu, I., J.~Krivine and D.~Varacca, \emph{A compositional semantics for
  the reversible p-calculus}, in: \emph{28th Annual {ACM/IEEE} Symposium on
  Logic in Computer Science, {LICS} 2013, New Orleans, LA, USA, June 25-28,
  2013}, pages 388--397, {IEEE} Computer Society (2013).
\newline\urlprefix\url{https://doi.org/10.1109/LICS.2013.45}

\bibitem{Cri_etal16}
Cristescu, I., J.~Krivine and D.~Varacca, \emph{Rigid families for the
  reversible {\(\pi\)}-calculus}, in: S.~J. Devitt and I.~Lanese, editors,
  \emph{Reversible Computation - 8th International Conference, {RC} 2016,
  Bologna, Italy, July 7-8, 2016, Proceedings}, volume 9720 of \emph{Lecture
  Notes in Computer Science}, pages 3--19, Springer (2016), ISBN
  978-3-319-40578-0.
\newline\urlprefix\url{https://doi.org/10.1007/978-3-319-40578-0_1}

\bibitem{DanKri04}
Danos, V. and J.~Krivine, \emph{Reversible communicating systems}, in:
  P.~Gardner and N.~Yoshida, editors, \emph{CONCUR 2004 - Concurrency Theory,
  15th International Conference, London, UK, August 31 - September 3, 2004,
  Proceedings}, volume 3170 of \emph{Lecture Notes in Computer Science}, pages
  292--307, Springer (2004), ISBN 3-540-22940-X.
\newline\urlprefix\url{https://doi.org/10.1007/978-3-540-28644-8_19}

\bibitem{DanKri05}
Danos, V. and J.~Krivine, \emph{Transactions in {RCCS}}, in: M.~Abadi and
  L.~de~Alfaro, editors, \emph{{CONCUR} 2005 - Concurrency Theory, 16th
  International Conference, {CONCUR} 2005, San Francisco, CA, USA, August
  23-26, 2005, Proceedings}, volume 3653 of \emph{Lecture Notes in Computer
  Science}, pages 398--412, Springer (2005), ISBN 3-540-28309-9.
\newline\urlprefix\url{https://doi.org/10.1007/11539452_31}

\bibitem{DanKriSob06}
Danos, V., J.~Krivine and P.~Soboci{\'n}ski, \emph{General reversibility}, in:
  R.~M. Amadio and I.~Phillips, editors, \emph{Proceedings of the 13th
  International Workshop on Expressiveness in Concurrency, {EXPRESS} 2006,
  Bonn, Germany, August 26, 2006}, volume 175 of \emph{Electronic Notes in
  Theoretical Computer Science}, pages 75--86, Elsevier (2006).
\newline\urlprefix\url{https://doi.org/10.1016/J.ENTCS.2006.07.036}

\bibitem{DeNMonVaa90}
{De Nicola}, R., U.~Montanari and F.~W. Vaandrager, \emph{Back and forth
  bisimulations}, in: J.~C.~M. Baeten and J.~W. Klop, editors, \emph{{CONCUR}
  '90, Theories of Concurrency: Unification and Extension, Amsterdam, The
  Netherlands, August 27-30, 1990, Proceedings}, volume 458 of \emph{Lecture
  Notes in Computer Science}, pages 152--165, Springer (1990).
\newline\urlprefix\url{https://doi.org/10.1007/BFB0039058}

\bibitem{Sup_etal16}
Dutta, S., B.~Adhikari and S.~Banerjee, \emph{Seidel switching for weighted
  multi-digraphs and its quantum perspective}, ArXiv: Combinatorics  (2016).
\newline\urlprefix\url{https://doi.org/10.48550/arXiv.1608.07830}

\bibitem{Graversen2018}
Graversen, E., I.~Phillips and N.~Yoshida, \emph{Event structure semantics of
  (controlled) reversible {CCS}}, in: J.~Kari and I.~Ulidowski, editors,
  \emph{Reversible Computation - 10th International Conference, {RC} 2018,
  Leicester, UK, September 12-14, 2018, Proceedings}, volume 11106 of
  \emph{Lecture Notes in Computer Science}, pages 102--122, Springer (2018).
\newline\urlprefix\url{https://doi.org/10.1007/978-3-319-99498-7_7}

\bibitem{Lan_etal21}
Lanese, I., D.~Medi{\'c} and C.~A. Mezzina, \emph{Static versus dynamic
  reversibility in {CCS}}, Acta Informatica  (2019).
\newline\urlprefix\url{https://doi.org/10.1007/s00236-019-00346-6}

\bibitem{Lan_etal11}
Lanese, I., C.~A. Mezzina, A.~Schmitt and J.~Stefani, \emph{Controlling
  reversibility in higher-order pi}, in: J.~Katoen and B.~K{\"o}nig, editors,
  \emph{{CONCUR} 2011 - Concurrency Theory - 22nd International Conference,
  {CONCUR} 2011, Aachen, Germany, September 6-9, 2011. Proceedings}, volume
  6901 of \emph{Lecture Notes in Computer Science}, pages 297--311, Springer,
  Berlin, Heidelberg (2011).
\newline\urlprefix\url{https://doi.org/10.1007/978-3-642-23217-6_20}

\bibitem{Lan_etal10}
Lanese, I., C.~A. Mezzina and J.-B. Stefani, \emph{Reversing higher-order pi},
  in: P.~Gastin and F.~Laroussinie, editors, \emph{CONCUR 2010 - Concurrency
  Theory, 21th International Conference, CONCUR 2010, Paris, France, August
  31-September 3, 2010. Proceedings}, volume 6269 of \emph{Lecture Notes in
  Computer Science}, pages 478--493, Springer (2010), ISBN 978-3-642-15374-7.
\newline\urlprefix\url{https://doi.org/10.1007/978-3-642-15375-4_33}

\bibitem{Lan_etal14}
Lanese, I., C.~A. Mezzina and F.~Tiezzi, \emph{Causal-consistent
  reversibility}, Bulletin of the European Association for Theoretical Computer
  Science \textbf{114}, page~17 (2014).

\bibitem{Lan_etal18}
Lanese, I., N.~Nishida, A.~Palacios and G.~Vidal, \emph{Cauder: {A}
  causal-consistent reversible debugger for erlang}, in: J.~P. Gallagher and
  M.~Sulzmann, editors, \emph{Functional and Logic Programming - 14th
  International Symposium, {FLOPS} 2018, Nagoya, Japan, May 9-11, 2018,
  Proceedings}, volume 10818 of \emph{Lecture Notes in Computer Science}, pages
  247--263, Springer, Cham (2018), ISBN 978-3-319-90686-7.
\newline\urlprefix\url{https://doi.org/10.1007/978-3-319-90686-7\_16}

\bibitem{Lanese2021}
Lanese, I. and I.~Phillips, \emph{Forward-reverse observational equivalences in
  {CCSK}}, in: S.~Yamashita and T.~Yokoyama, editors, \emph{Reversible
  Computation - 13th International Conference, {RC} 2021, Virtual Event, July
  7-8, 2021, Proceedings}, volume 12805 of \emph{Lecture Notes in Computer
  Science}, pages 126--143, Springer (2021), ISBN 978-3-030-79836-9.
\newline\urlprefix\url{https://doi.org/10.1007/978-3-030-79837-6_8}

\bibitem{LPU24}
Lanese, I., I.~Phillips and I.~Ulidowski, \emph{An axiomatic theory for
  reversible computation}, ACM Transactions on Computational Logic \textbf{25},
  pages 1--40 (2024).
\newline\urlprefix\url{https://doi.org/10.1145/3648474}

\bibitem{Lie_etal12}
Lienhardt, M., I.~Lanese, C.~A. Mezzina and J.~Stefani, \emph{A reversible
  abstract machine and its space overhead}, in: H.~Giese and G.~Rosu, editors,
  \emph{Formal Techniques for Distributed Systems - Joint 14th {IFIP} {WG} 6.1
  International Conference, {FMOODS} 2012 and 32nd {IFIP} {WG} 6.1
  International Conference, {FORTE} 2012, Stockholm, Sweden, June 13-16, 2012.
  Proceedings}, volume 7273 of \emph{Lecture Notes in Computer Science}, pages
  1--17, Springer, Berlin, Heidelberg (2012), ISBN 978-3-642-30793-5.
\newline\urlprefix\url{https://doi.org/10.1007/978-3-642-30793-5_1}

\bibitem{Mel_etal21}
Melgratti, H., C.~A. Mezzina and G.~Michele~Pinna, \emph{A distributed
  operational view of reversible prime event structures}, in: \emph{2021 36th
  Annual ACM/IEEE Symposium on Logic in Computer Science (LICS)}, pages 1--13
  (2021).
\newline\urlprefix\url{https://doi.org/10.1109/LICS52264.2021.9470623}

\bibitem{lmcs:12332}
Melgratti, H., C.~A. Mezzina and G.~M. Pinna, \emph{A truly concurrent
  semantics for reversible ccs}, Logical Methods in Computer Science
  \textbf{Volume 20, Issue 4}, 20 (2024), ISSN 1860-5974.
\newline\urlprefix\url{https://doi.org/10.46298/lmcs-20(4:20)2024}

\bibitem{Mil89}
Milner, R., \emph{Communication and Concurrency}, {PHI} Series in computer
  science, Prentice-Hall (1989), ISBN 978-0-13-115007-2.

\bibitem{MonPis97}
Montanari, U. and M.~Pistore, \emph{Minimal transition systems for
  history-preserving bisimulation}, in: R.~Reischuk and M.~Morvan, editors,
  \emph{{STACS} 97, 14th Annual Symposium on Theoretical Aspects of Computer
  Science, L{\"{u}}beck, Germany, February 27 - March 1, 1997, Proceedings},
  volume 1200 of \emph{Lecture Notes in Computer Science}, pages 413--425,
  Springer (1997).
\newline\urlprefix\url{https://doi.org/10.1007/BFB0023477}

\bibitem{NiePloWin81}
Nielsen, M., G.~D. Plotkin and G.~Winskel, \emph{Petri nets, event structures
  and domains, part {I}}, Theoretical Computer Science \textbf{13}, pages
  85--108 (1981).
\newline\urlprefix\url{https://doi.org/10.1016/0304-3975(81)90112-2}

\bibitem{Pac_etal13}
Pach, P.~P., M.~Pinsker, A.~Pongr{\'a}cz and C.~Szab{\'o}, \emph{A new
  operation on partially ordered sets}, Journal of Combinatorial Theory, Series
  A \textbf{120}, pages 1450--1462 (2013), ISSN 0097-3165.
\newline\urlprefix\url{https://doi.org/https://doi.org/10.1016/j.jcta.2013.04.003}

\bibitem{PhiUli06}
Phillips, I. and I.~Ulidowski, \emph{Reversing algebraic process calculi}, in:
  L.~Aceto and A.~Ing{\'{o}}lfsd{\'{o}}ttir, editors, \emph{Foundations of
  Software Science and Computation Structures, 9th International Conference,
  {FOSSACS} 2006, Held as Part of the Joint European Conferences on Theory and
  Practice of Software, {ETAPS} 2006, Vienna, Austria, March 25-31, 2006,
  Proceedings}, volume 3921 of \emph{Lecture Notes in Computer Science}, pages
  246--260, Springer (2006), ISBN 3-540-33045-3.
\newline\urlprefix\url{https://doi.org/10.1007/11690634_17}

\bibitem{PhiUli12}
Phillips, I. and I.~Ulidowski, \emph{A hierarchy of reverse bisimulations on
  stable configuration structures}, Mathematical Structures in Computer Science
  \textbf{22}, pages 333--372 (2012), ISSN 0960-1295.
\newline\urlprefix\url{https://doi.org/10.1017/S0960129511000429}

\bibitem{PhiUli15}
Phillips, I. and I.~Ulidowski, \emph{Reversibility and asymmetric conflict in
  event structures}, Journal of Logical and Algebraic Methods in Programming
  \textbf{84}, pages 781--805 (2015), ISSN 2352-2208. ``Special Issue on Open
  Problems in Concurrency Theory''.
\newline\urlprefix\url{https://doi.org/10.1016/J.JLAMP.2015.07.004}

\bibitem{Pra87}
Prasad, K., \emph{Combinators and bisimulation proofs for restartable systems},
  Ph.D. thesis, University of Edinburgh (1987).

\bibitem{Seidel73}
Seidel, J.~J., \emph{A survey of two-graphs}, Geometry and Combinatorics pages
  146--176 (1991).
\newline\urlprefix\url{https://doi.org/10.1016/B978-0-12-189420-7.50018-9}

\bibitem{Sel08}
Selinger, P., \emph{Lecture notes on the lambda calculus}, arXiv preprint
  arXiv:0804.3434  (2008).

\bibitem{She_etal03}
Shende, V.~V., A.~K. Prasad, I.~L. Markov and J.~P. Hayes, \emph{Synthesis of
  reversible logic circuits}, {IEEE} Transactions on Computer-Aided Design of
  Integrated Circuits and Systems \textbf{22}, pages 710--722 (2003).
\newline\urlprefix\url{https://doi.org/10.1109/TCAD.2003.811448}

\bibitem{rc2020}
Ulidowski, I., I.~Lanese, U.~P. Schultz and C.~Ferreira, editors,
  \emph{Reversible Computation: Extending Horizons of Computing - Selected
  Results of the {COST} Action {IC1405}}, volume 12070 of \emph{Lecture Notes
  in Computer Science}, Springer (2020), ISBN 978-3-030-47360-0.
\newline\urlprefix\url{https://doi.org/10.1007/978-3-030-47361-7}

\bibitem{Uli_etal18}
Ulidowski, I., I.~Phillips and S.~Yuen, \emph{Reversing event structures}, New
  Generation Computing \textbf{36}, pages 281--306 (2018).
\newline\urlprefix\url{https://doi.org/10.1007/S00354-018-0040-8}

\bibitem{GlaPlo95}
van Glabbeek, R.~J. and G.~D. Plotkin, \emph{Configuration structures}, in:
  \emph{Proceedings, 10th Annual {IEEE} Symposium on Logic in Computer Science,
  San Diego, California, USA, June 26-29, 1995}, pages 199--209, {IEEE}
  Computer Society (1995).
\newline\urlprefix\url{https://doi.org/10.1109/LICS.1995.523257}

\bibitem{GlaPlo09}
van Glabbeek, R.~J. and G.~D. Plotkin, \emph{Configuration structures, event
  structures and petri nets}, Theoretical Computer Science \textbf{410}, pages
  4111--4159 (2009).
\newline\urlprefix\url{https://doi.org/10.1016/j.tcs.2009.06.014}

\bibitem{Win82}
Winskel, G., \emph{Event structure semantics for {CCS} and related languages},
  in: M.~Nielsen and E.~M. Schmidt, editors, \emph{Automata, Languages and
  Programming, 9th Colloquium, Aarhus, Denmark, July 12-16, 1982, Proceedings},
  volume 140 of \emph{Lecture Notes in Computer Science}, pages 561--576,
  Springer (1982), ISBN 3-540-11576-5.
\newline\urlprefix\url{https://doi.org/10.1007/BFb0012800}

\end{thebibliography}

\end{document}